\documentclass[11pt]{article}
\usepackage[margin=1.1in]{geometry}
\usepackage{amsmath,amssymb,amsthm,mathtools}
\usepackage{etoolbox}
\usepackage[T1]{fontenc}
\usepackage{microtype}
\usepackage{booktabs}
\usepackage{array}
\usepackage{longtable}
\usepackage{pdflscape}
\usepackage{graphicx}
\usepackage{xcolor}
\usepackage{algorithm}
\usepackage{algpseudocode}
\usepackage{enumitem}
\usepackage[authoryear]{natbib}
\usepackage[colorlinks=true,linkcolor=black,citecolor=black,urlcolor=black]{hyperref}

\newcommand{\R}{\mathbb R}
\newcommand{\E}{\mathbb E}
\newcommand{\Q}{\mathbb Q}
\newcommand{\PP}{\mathbb P}
\newcommand{\dd}{\mathrm d}

\newcommand{\abs}[1]{\lvert #1\rvert}

\newcommand{\one}{\mathbf 1}

\newcommand{\removeblock}[1]{}

\DeclareMathOperator{\diag}{diag}

\renewrobustcmd{\P}{\mathbb{P}}

\theoremstyle{plain}
\newtheorem{theorem}{Theorem}[section]

\newtheorem{lemma}[theorem]{Lemma}
\newtheorem{corollary}[theorem]{Corollary}
\theoremstyle{definition}
\newtheorem{assumption}[theorem]{Assumption}
\theoremstyle{remark}
\newtheorem{remark}[theorem]{Remark}

\title{%
  \textbf{GROMM}\\[6pt]
  \large Generative Risk-neutral Option Market Models
}

\author{Magnus Wiese\qquad Hans Buehler$^{*}$\qquad Blanka Horvath$^{*,\#}$\\[0.6em]
  {\small $^{*}$Mathematical Institute, University of Oxford, 
  $^{\#}$Oxford-Man Institute, University of Oxford}}
\date{Version 1.1\\ \today}

\begin{document}
\maketitle

\begin{abstract}
    We present a statically and dynamically arbitrage-free European option market model with smooth
    interpolation between option prices, based on the SANOS~\citep{SANOS2026} framework, and prove that its dynamics admit a unique strong solution. 

    In spirit,
    it is a stochastic model for local volatility when the latter is used as a representation of a European option price surface, not as a pricing model.

  We illustrate the strength of our approach on S\&P Index Option data from January 2022 to August 2025.
  We show that the model's initial state fits within bid-ask for more than 95\% of relevant market option prices.\footnote{
        The model presented was trained for daily expiries up to 25 business days and strikes between $[\exp(-2.5 \sigma_{atm} \sqrt{t}), \exp(+1.5 \sigma_{atm} \sqrt{t})]$.}
    We then train the model dynamics under~$\PP$ and~$\Q$ and show that the model is indeed free of arbitrage
    while fitting the dynamics of the market well.
    We also propose a range of BHW-tests to test for absence of arbitrage. 
\end{abstract}

\tableofcontents

\section{Introduction}\label{sec:intro}

This article presents an arbitrage-free market model for smooth surfaces of European option prices. While not strictly the first of its kind, it is the first model with an industry-grade ability to fit the current market: GROMM fits thousands of short-dated options per day, including the next-day expiry which most models struggle to fit, with more than 95\% of the calibrated prices lying within the bid-ask spread, whilst making no parametric assumptions on the modelled surface. In addition, GROMM allows for smooth interpolation of the entire option surface at each time, provides an explicit construction of the continuous-time risk-neutral drift for both (discrete local) volatilities and spot, and also has a natural estimation strategy for its parameter space.

The core architecture follows the SANOS approach~\citep{SANOS2026}, which models the stock price as the product
of a time-and-state discrete process~$X$ and a smoothing process~$Y$, in our case Black~\&~Scholes.
The former is used to encode the option market using a discrete form of local volatilities as code book. That does not mean
that the model is a local volatility model -- we merely use the implied local volatilities as a convenient representation of the option surface.
Compared to using implied volatilities the advantage is that mere positivity of the local volatilities is sufficient to ensure
the option market is free of static arbitrage at all times. We have shown in~\cite{SANOS2026} that this representation
can be fitted to real life market data very efficiently using linear programming.

The key contribution in this paper is showing how to construct a dynamic model of SANOS option call prices which are not only statically but
also dynamically free of arbitrage. We show the theoretical construction of the corresponding risk-neutral measure~$\mathbb{Q}$
and prove that its dynamics admit a unique strong solution. We then discuss implementation
and estimate our model on S\&P Index Option data from January 2022 to August 2025. We fit our daily model parameters to observed bid-ask prices and then proceed to
estimate both $\Q$ and its drift under $\PP$. We show that
on this data set more than 95\% of calibrated
model prices fit within observed market bid-ask spreads even with the simplistic grid architecture we present here.

We then show that model dynamics under the risk-neutral measure indeed have numerically zero conditional and unconditional drift. We test this first
with a restricted universe of 570 tradable instruments over several time steps in a number of different set-ups. We then also show that even using all 5{,}001 instruments
 risk-neutrality is preserved. Ours is the first model being able to show such a result for a number of instruments which far exceeds
the number of driving stochastic factors.

The model presented here is distinct from our own DYSANOS~\citep{DYSANOS2026}. There, we focused on a re-parametrization of SANOS into machine-learnable parameters
and then on learning
realistic market dynamics for smooth statically arbitrage-free SANOS
surfaces under~$\PP$. Existence of a risk-neutral measure was secondary.\footnote{
We did also construct an equivalent risk-neutral measure~$\Q$ numerically using the \emph{martingale measure} approach introduced in~\cite{removedrift}.
While this method is model-free and can in theory be applied to any option market model under~$\PP$, it is numerically heavy and in practice
restricted to making only a fairly limited number of the overall tradable instruments
dynamically arbitrage-free: namely, in~\cite{DYSANOS2026} we used only 49 instruments per simulated date.}
With GROMM we present a framework whose foremost goal is the construction of a risk-neutral measure for \emph{all} options with
an explicit construction of the risk-neutral drift. The price we pay is a reduction
in the parameter space available for fitting historic data compared to DYSANOS.

\subsubsection*{Related Work}
Modelling joint stock--option dynamics requires fitting market quotes, preserving static arbitrage freedom across strikes and maturities, and excluding dynamic trading arbitrage. These are distinct objectives: a sequence of admissible surfaces need not admit a common equivalent martingale measure. Early empirical studies identified low-dimensional patterns in implied-volatility movements~\citep{SkiadopoulosHodgesClewlow2000,ContDaFonseca2002Dynamics,FenglerHardleVilla2003} and were followed by factor models~\citep{ContVol,Hafner2004,FenglerHardleMammen2007}. However, such factors do not automatically preserve admissible option prices. 

No-arbitrage drift restrictions for stochastic implied volatilities go back to~\cite{Schonbucher1999} and~\cite{ledoit2002relative}. Corresponding existence results were first presented for a fixed-strike term structure in~\cite{SchweizerWisselTerm2008} and for the multi-strike case in~\cite{SchweizerWissel2008}. The symmetric SSVI specification in~\cite{DysLitElAmrani2021} is an instructive example that letting the parameters of a statically arbitrage-free family of surfaces diffuse does not yield arbitrage-free dynamics. 

Our approach is most closely related to the seminal but unpublished work of~\cite{Wissel2007} which extended his earlier works~\citep{SchweizerWissel2008, SchweizerWisselTerm2008}. From a theoretical perspective, our approach is a version of~\cite{CarmonaNadtochiy2009} where the authors derive related risk-neutral drift conditions for continuous time local volatility models, but without a tangible implementation strategy. Our results can be seen as a joint extension of the two papers, based on the statically arbitrage-free DYSANOS~\citep{DYSANOS2026}. Its core parameters are smoothed \emph{discrete local volatilities} introduced in~\cite{SANOS2026} based on earlier work in~\cite{AndreasenHuge, DLV}.

In~\cite{TangentLevy} the authors solve the problem of finding a dynamically arbitrage-free evolution of the full option surface via products of independent L\'evy processes. Necessary and sufficient no-arbitrage conditions for time-inhomogeneous L\'evy exponents within an HJM framework are supplied in~\cite{KallsenKruhner2015}, and Tangent L\'evy models are estimated and simulated in~\cite{CarmonaMaNadtochiy2017}. 
While structurally
very elegant, option prices generated by products of independent increment processes do not typically fit well to equity index European option surfaces as we will illustrate below. In contrast, our model prices can be fit very efficiently within nearly all\footnote{In the experiments below, above 95\%.} observed market bid-ask prices using the linear programming SANOS~\citep{DLV, SANOS2026} calibration method. 

In \cite{DysLitGambaraTeichmann2021} the authors propose a piecewise
dynamically arbitrage-free model via ``regridding''.  Below we implement their
generalized-Bates example and submit its numerical realization to the same
unconditional and conditional diagnostics as GROMM and Tangent L\'evy. Our trained version of this
model exhibits dynamic arbitrage.

Forward-variance models~\citep{DysLitBergomi2005,bergomi2008smile}, for which~\cite{buehler2006consistent} derives the HJM-type no-arbitrage conditions, stochastic-volatility smile expansions~\citep{bergomi2012stochastic} and leverage-based analysis~\citep{DysLitCilibertiBouchaudPotters2008} investigate the volatility structures underlying smile dynamics. These provide economic targets for surface simulators, rather than the same nonparametric surface code book. Rough-SABR and asymptotically arbitrage-free approximation results~\citep{DysLitFukasawaGatheral2022,DysLitFukasawa2022} and normalizing volatility transforms~\citep{lucic2021normalizing} connect implied-volatility approximations and smile parametrizations to static no-arbitrage.

\subsubsection*{Generative Models under the Statistical Measure}

In order to have strict absence of dynamic arbitrage a model must also be free of \emph{static arbitrage}.
There are few works on generative equity option market dynamics which satisfy this requirement -- apart from the above quoted works, the main
examples are \cite{wiese2019simulate}, the multi-asset version~\citep{wiese2021multiassetspotoptionmarket} and
the neural-SDE based model~\citep{Cohen02012023}. These models generate
discrete strictly arbitrage-free option surfaces for a limited number of options. Path-dependent SSVI dynamics~\citep{AndresBoumezouedJourdain2023} and VAEs decoding into arbitrage-free pricing-model parameters~\citep{vae_ning2023arbitrage} also preserve static admissibility at simulated dates by employing parametric surfaces. None of these works introduced a risk-neutral measure. To address the absence of an analytical risk-neutral measure, we presented in~\cite{removedrift} a model-free approach to learning such a measure. However, that method is numerically challenging for a large number of tradable options as discussed in~\cite{DYSANOS2026}.

Most other approaches to generative option markets such as~\cite{BergeronEtAl2021,ZhangLiZhang2021,ContVuletic2023,ChoudharyJaimungalBergeron2024,VuleticCont2024,JinAgarwal2025,BrooksEtAl2026,HaoJi2026,HanEtAl2026,YangWangZhangLi2026,BucheggerGonon2026} are not strictly free of arbitrage. In their works price admissibility is only encouraged by penalties and reweighting and price corrections. Moreover, they do not construct corresponding risk-neutral dynamics. Table~\ref{tab:surface-dynamics-literature}
on page~\pageref{tab:surface-dynamics-literature} provides a summary of the literature on the topic.

For interest rate option markets the reference model in the same class is the BGM Model~\citep{BGM}. 

\section{Stochastic Smooth Discrete Local Volatility}

In this section we introduce the core model architecture for our \emph{Stochastic 
Smooth Discrete Local Volatility model}.
We follow the core SANOS~\citep{SANOS2026} idea of writing our equity process $S$ as a product
\begin{equation}
 \label{eq:S} 
  S_t = X_t Y_t 
\end{equation}
of a \emph{Stochastic Discrete Local Volatility (SDLV)} process $X$ which is used to fit to current market conditions and 
whose state parameters will have consistent $\PP$ and $\Q$ dynamics, and an independent\footnote{Under both $\PP$ and $\Q$.} and rather trivial smoothing process $Y$ whose dynamics under $\PP$ and $\Q$ are the same.

The process $Y$ will just be a Black~\&~Scholes process with fixed volatility~5\%. Given observed bid-ask market data it was shown in~\cite{SANOS2026} how to fit
a discrete local volatility surface in the presence of such a process. 

In the remainder of this section we focus on the SDLV process itself.

\subsection{Stochastic Discrete Local Volatility (SDLV)}

Our \emph{Stochastic Discrete Local Volatility (SDLV)} risk-neutral model will be defined on a \emph{fixed} expiry and cash strike grid. It will not use a floating grid in its option market dynamics.\footnote{This is not to be confused with 
how the model parameters are estimated using historic market data. We explain the estimation in Section~\ref{sec:estimation}.} We will use ``$t$'' to refer to 
model time. Denote initial spot by~$X_0>0$.

For now let $0=T_0<T_1<\dots<T_M$ be the fixed model expiry dates, and assume that $K^0<K^1<\cdots<X_0<\cdots<K^{N+1}$ are the fixed model cash strikes, where 
$K^0$ and $K^{N+1}$
are boundary strikes which are far away from the inner strikes, cf.~\cite{SANOS2026}.\footnote{In our experiments the inner strikes are given as a range of strikes
from $K^1=\exp(-2.5 \sigma_{atm}(T_M) \sqrt{T_M})$ to $K^N=\exp(+1.5 \sigma_{atm}(T_M) \sqrt{T_M})$ with boundaries set to $K^0=K^1/2$ and $K^{N+1}=3 K^N / 2$; see
Section~\ref{sec:estimation-data} for specifics.
}

Assume that $t$ is the current model time and let $T_k$ be the front expiry such that $T_{k-1}\leq t<T_k$.
Our model is given at any time~$t$
by the current spot $X_t$ and future densities $q_t^j$, where $q_t^{j,i}$ represents the risk-neutral density of the $i$-th cash strike at future expiry $T_j$.

Such discrete densities have been shown in~\cite{DLV} to admit an equivalent representation via (squared) \emph{discrete local volatilities} $\Sigma_t^{j,1},\ldots,\Sigma_t^{j,N}$ in the following form:
there is a tridiagonal matrix~$\Omega_j$\footnote{
  The matrix $\Omega_j\in\mathbb R^{(N+2)\times(N+2)}$ is given as follows: 
  let $\gamma^{\pm i} := \pm (K^{i})^2 (T_j-T_{j-1}) / (K^{i+1}-K^{i-1}) / (K^{i\pm 1} - K^i) $ for $i=1,\ldots,N$.
  Then $\Omega_j$ is the tridiagonal matrix whose $i$-th column, $i=1,\ldots,N$, has the entries
  \[
    (\Omega_j)_{i-1,i} = \gamma^{-i},\quad
    (\Omega_j)_{i,i} = -(\gamma^{-i}+\gamma^{+i}),\quad
    (\Omega_j)_{i+1,i} = \gamma^{+i},
  \]
  and all other entries zero, so that $\one^{\!\top}\Omega_j=0$ and $K^{\!\top}\Omega_j=0$;
  cf.~\cite{SANOS2026}.
 } for each $j>k$ such that
    \begin{equation}
      q_t^{j-1}=M_t^jq_t^j \ ,
      \qquad
      M_t^j:=I-\Omega_jS_t^j \ ,
      \qquad
      S_t^j:=\operatorname{diag}(0,\Sigma_t^{j,1},\ldots,
        \Sigma_t^{j,N},0) 
      \label{eq:hb_dlv}
    \end{equation}
where $M^j_t$ is an invertible tridiagonal M-matrix for every non-negative $S_t^j$. The key point is that $q_t^j$ for $j>k$ are uniquely determined martingale densities
for any given initial density $q_t^k$ and the surface of non-negative (squared) discrete local vols.

  \subsection{Risk-Neutral Dynamics}

Under a risk-neutral measure~$\Q$ densities are martingales. We will now turn this property into conditions on dynamics of the (squared) discrete local volatilities, and, 
jointly the front density and spot. 
      Our approach follows in spirit that of~\cite{Wissel2007} in the sense that we rely on the martingale property of a discrete set of tradable derivative assets.
      However, we reframe it in a much more accessible coordinate system.

Assume that $W$ is a $d$-dimensional Brownian motion. We designate $W^1$ as the
      spot factor: spot is always driven by $dW_t^1$.

      \subsubsection{Future Densities}
  \label{sec:future-densities}

       Each density coordinate
      is a martingale. We therefore specify the front-density diffusion as
      \[
        dq_t^k=E_t^{k;1}\,dW_t^1 \ ,
        \qquad
        E_t^{k;1}:=\operatorname{diag}(q_t^k)\eta_t^{k;1} \ ,
        \qquad E_t^{k;f}:=0\quad(f=2,\ldots,d),
      \]
      where $E_t^{k;1}\in\mathbb R^{N+2}$ must satisfy
      $\one^{\!\top}E_t^{k;1}=0$, so that $q_t^k$ retains unit mass. The
      remaining Brownian factors do not directly move the front density. They
      drive the stochastic DLV surface, while $\xi_t^{j,\cdot;1}$ determines their correlation with spot.
      Future density volatilities $\eta_t^j$ for $j>k$ will be obtained recursively.

      For future expiries $T_j$ with $j>k$, assume the discrete local volatility squared is given as
      \begin{equation}
        \label{eq:dlv_sde_hb}
        \frac{ d\Sigma_t^{j} }{ \Sigma_t^{j} }
          = \mu^j_t\,dt
          + \xi^j_t\,dW_t,
      \end{equation}
      where $\xi_t^j\in\mathbb R^{N\times d}$ is given and
      $\mu_t^j\in\mathbb R^N$ is the unknown risk-neutral drift. We write this as
      \[
        dS_t^j=A_t^j\,dt+\sum_{f=1}^dB_t^{j;f}\,dW_t^f,
      \]
      with
        $A_t^j
        :=\operatorname{diag}(0,\Sigma_t^{j,1}\mu_t^{j,1},\ldots,
            \Sigma_t^{j,N}\mu_t^{j,N},0)$ and $B_t^{j;f}
        :=\operatorname{diag}(0,\Sigma_t^{j,1}\xi_t^{j,1;f},\ldots,
            \Sigma_t^{j,N}\xi_t^{j,N;f},0).$
      With the above
      \[
        dM_t^j=-\Omega_j\,dS_t^j
        =-\Omega_jA_t^j\,dt
        -\sum_{f=1}^d\Omega_jB_t^{j;f}\,dW_t^f.
      \]
      Setting $dq_t^j=\sum_fE_t^{j;f}\,dW_t^f$ for an unknown $E_t^{j;f}$ yields
      \begin{align}
        dq_t^{j-1}
        &= M_t^j\,dq_t^j+(dM_t^j)q_t^j+d[M^j,q^j]_t\\
        &= \sum_{f=1}^d
          \left(M_t^jE_t^{j;f}-\Omega_jB_t^{j;f}q_t^j\right)dW_t^f
          -\Omega_j\left(A_t^jq_t^j
          +\sum_{f=1}^dB_t^{j;f}E_t^{j;f}\right)dt \ .
      \end{align}
      Matching the diffusion coefficients gives
      \begin{equation}
        \label{eq:vol_hb}
        E_t^{j;f}
        =(M_t^j)^{-1}
          \left(E_t^{j-1;f}+\Omega_jB_t^{j;f}q_t^j\right)
        \in\mathbb R^{N+2} \ \ \ \mbox{and} \ \ \ \eta_t^{j,i;f}=\frac{ E_t^{j,i;f} }{ q_t^{j,i} } \ .
      \end{equation}
      The drift must vanish, hence
      \[
        \Omega_j\left(A_t^jq_t^j
        +\sum_{f=1}^dB_t^{j;f}E_t^{j;f}\right)=0.
      \]
      The restriction of $\Omega_j$ to vectors supported on the interior nodes
      is injective. Therefore the bracket vanishes node by node, giving
      \begin{equation}
        \label{eq:drift_hb}
        \mu_t^{j,i}
        =-\sum_{f=1}^d\xi_t^{j,i;f}
          \frac{E_t^{j,i;f}}{q_t^{j,i}}
        =-\sum_{f=1}^d\xi_t^{j,i;f}\eta_t^{j,i;f},
        \qquad i=1,\ldots,N.
      \end{equation}
      Thus the prescribed front-density volatility and the prescribed DLV
      volatilities determine all future density volatilities and the
      risk-neutral DLV drifts by forward induction over the expiries.

  \subsubsection{Front Density and Spot}
  \label{sec:front-density}

  The remaining task is to specify both spot dynamics consistent with the front density and the respective front density volatilities.

  \subsubsection*{Case $t=T_k-dt$}

  Assume we are at $t=T_k-dt$ and observe the front density $q_{t}^k$. Current spot is given as $X_t = K^\top q_t^k$.
  In order to sample a spot consistent with the discrete set of strikes, we can either draw a discrete sample or fit some other distribution exactly to the observed option prices (fitting to densities or probabilities will not work in general).

  The simplest method is to draw spot from this discrete distribution using
  the spot factor $W^1$. Let $F_t^k$ denote the conditional CDF which assigns
  probability $q_t^{k,i}$ to strike $K^i$, and let $\Phi$ denote the standard
  normal CDF. Then
  \[
    X_{T_k}=(F_t^k)^{-1}\!\left(\Phi(Z)\right),
    \qquad
    Z:=\frac{W_{T_k}^1-W_t^1}{\sqrt{dt}}\sim N(0,1).
  \]
  Correlation between spot and a future DLV coordinate is specified by that
  coordinate's loading $\xi_t^{j,i;1}$ on $W^1$.

  Conditional on this definition of terminal spot, the front-density
  diffusion is exact if the Brownian increment is revealed progressively:
  define
  \[
    Q_i:=\sum_{\ell=0}^i q_t^{k,\ell},
    \qquad
    a_i:=\Phi^{-1}(Q_i),
    \qquad
    a_{-1}:=-\infty,\quad a_{N+1}:=+\infty,
  \]
  and let $\phi$ be the standard normal density. For $t\leq s<T_k$, set
  \[
    z_i(s):=
      \frac{\sqrt{dt}\,a_i-(W_s^1-W_t^1)}{\sqrt{T_k-s}}.
  \]
  The time-$s$ front density is the conditional probability
  \begin{equation}
    \label{eq:front_density_mw}
    q_s^{k,i}
      :=\Q\!\left(X_{T_k}=K^i\mid\mathcal F_s\right)
      =\Phi\!\left(z_i(s)\right)
       -\Phi\!\left(z_{i-1}(s)\right).
  \end{equation}
  It is therefore a martingale. Applying It\^o's formula gives the exact
  diffusion
  \begin{equation}
    \label{eq:front_absolute_vol_mw}
    dq_s^{k,i}=E_s^{k,i;1}\,dW_s^1,
    \qquad
    E_s^{k,i;1}
      =\frac{\phi(z_{i-1}(s))-\phi(z_i(s))}{\sqrt{T_k-s}},
    \qquad t\leq s<T_k,
  \end{equation}
  with zero drift. Since $z_i(t)=a_i$, its initial value is
  \begin{equation}
    E_t^{k,i;1}
      =\frac{\phi(a_{i-1})-\phi(a_i)}{\sqrt{dt}} \ \ \ \mbox{hence} \ \ \
    \eta_t^{k,i;1}
      =\frac{\phi(a_{i-1})-\phi(a_i)}{q_t^{k,i}\sqrt{dt}}
      \ \ \ \mbox{for $i=0,\ldots,N+1$}
  \end{equation}
  while $E_t^{k,i;f}=\eta_t^{k,i;f}=0$ for $f=2,\ldots,d$.
  More generally, $\eta_s^{k,i;1}=E_s^{k,i;1}/q_s^{k,i}$ and the other
  factor loadings remain zero. Consequently,
  \[
    X_s=K^\top q_s^k
      =\mathbb E\!\left[X_{T_k}\mid\mathcal F_s\right],
    \qquad
    dX_s=K^\top E_s^{k;1}\,dW_s^1,
  \]
  so spot and the front density are driven by the same Brownian factor. At
  expiry, $q_s^k$ converges to the point mass associated with the sampled
  strike. The displayed coefficient $E_t^{k,i;1}$ is only the initial value
  of the exact, state- and time-dependent coefficient $E_s^{k,i;1}$; holding
  it constant over the whole interval would instead be an Euler
  approximation. Likewise, sampling $X_{T_k}$ all at once without
  progressively revealing $W^1$ produces a direct jump to the terminal point
  mass rather than a continuous density diffusion.

  \subsubsection*{Case $t\ll T_k$}

  The same construction gives a natural exact transition over any step which
  does not cross the front expiry. Let
  \[
    h:=T_k-t,
    \qquad 0<\delta\leq h,
    \qquad \Delta W_t^1:=W_{t+\delta}^1-W_t^1.
  \]
  Starting from $q_t^k$, define $Q_i$ and $a_i$ as above. After drawing
  $\Delta W_t^1\sim N(0,\delta)$, set
  \begin{equation}
    \label{eq:front-exact-step}
    q_{t+\delta}^{k,i}
      =\Phi\!\left(
          \frac{\sqrt h\,a_i-\Delta W_t^1}{\sqrt{h-\delta}}
        \right)
       -\Phi\!\left(
          \frac{\sqrt h\,a_{i-1}-\Delta W_t^1}{\sqrt{h-\delta}}
        \right),
    \qquad \delta<h.
  \end{equation}
  The corresponding spot is not sampled independently but is defined by
  \begin{equation}
    \label{eq:spot-exact-step}
    X_{t+\delta}:=K^\top q_{t+\delta}^k
      =\mathbb E\!\left[X_{T_k}\mid\mathcal F_{t+\delta}\right].
  \end{equation}
  Consequently
  \[
    \mathbb E_t[q_{t+\delta}^k]=q_t^k,
    \qquad
    \mathbb E_t[X_{t+\delta}]=X_t
  \]
  exactly. In general $X_{t+\delta}$ is not a grid strike; only at $T_k$ does
  the density collapse to a unit vector and spot become the corresponding
  strike.

  This update is recursively consistent. Indeed, the Gaussian boundaries
  calculated from the cumulative probabilities of $q_{t+\delta}^k$ are
  precisely the residual boundaries appearing in
  \eqref{eq:front-exact-step}. The construction may therefore be restarted
  from $q_{t+\delta}^k$ without changing the joint law. It is a
  simplex-valued information bridge: conditional on the exact terminal
  Gaussian variable, $W^1$ is a Brownian bridge; conditional only on the
  terminal strike, it is Brownian motion conditioned to terminate in an
  interval, or equivalently a mixture of Brownian bridges.

\subsubsection{Existence}
\label{sec:existence}
We have just described a set of risk-neutral model dynamics for a discrete process $X_t$ whose future densities themselves are martingales driven by stochastic discrete local volatilities. Our construction is concise and lends itself to efficient numerical implementation and model parameter estimation. We now turn to proving the existence of the process. 

Let $L_t^{j,i}\coloneqq\log\Sigma_t^{j,i}$ and let $B_t = W_t^1 - W_{T_{k-1}}^1$ be the Brownian motion centred at the origin. Given \eqref{eq:dlv_sde_hb} and the drift restriction \eqref{eq:drift_hb}, the It\^o dynamics of the log of the DLV squared are given by 
\begin{equation}
  \label{eq:log_dlv_sde_mw}
  dL_t^{j,i}=\left(-\sum_{f=1}^d\xi_t^{j,i;f}\eta_t^{j,i;f}
    -\frac12\sum_{f=1}^d(\xi_t^{j,i;f})^2\right)dt+\sum_{f=1}^d\xi_t^{j,i;f}\,dW_t^f 
\end{equation}
where $\xi_t^{j,i;f} = \xi^{j,i;f}(t, B_t, L_t)$ and $\eta_t^{j,i;f} = \eta^{j,i;f}(t, B_t, L_t)$.

Upon inspection of \eqref{eq:log_dlv_sde_mw}, the drift restriction can cause instabilities that call the existence of the process into question: as time approaches the front's expiry $T_k$, the densities of the future expiries lose mass on the wing nodes the front abandons, which causes $\eta^{j, i; 1} = E^{j, i; 1} / q^{j, i}$ to blow up. 
Moreover, for the Gaussian bridge front as introduced in Section~\ref{sec:front-density}, $\eta_t^{k}$ grows asymptotically at the rate $(T_k - t)^{-1}$ on nodes it abandons\footnote{
Let $K^i$ be a node that is abandoned at $T_k$ and $K^i < K^J$ without loss of generality where $K^J$ is the node at which spot settles. Let $\varepsilon_i$ denote the distance to the bucket $a_i$: 
$$ \varepsilon_i \coloneqq (W^1_s - W^1_{T_{k-1}}) - \sqrt{T_k - T_{k-1}} \, a_i > 0 \, . $$ 
We obtain $z_i = -\varepsilon_i / \sqrt{\tau}\to -\infty$ as $\tau \coloneqq T_k - s \to 0$; similar holds for $z_{i-1} < z_i$ which is even more negative. 
Dividing \eqref{eq:front_absolute_vol_mw} by \eqref{eq:front_density_mw} we obtain 
\begin{equation*}
  \eta^{k, i; 1} = \dfrac{\phi(z_{i-1}) - \phi(z_i)}{(\Phi(z_i) - \Phi(z_{i-1}))  \sqrt{T_k - s}} \, . 
\end{equation*}
As $z_i, z_{i-1} \to - \infty$, the terms in $z_{i-1}$ become negligible for both numerator and denominator. Moreover, $\Phi(-x) = 1-\Phi(x) \sim \phi(x)/x$ as $x \to \infty$ 
\citep[Ch.~VII, Lemma~2, p.~175]{feller-vol-1}. Combining both we obtain 
\begin{equation*}
  \eta^{k, i; 1} \sim \frac{-\phi(z_i)}{\phi(z_i)\,(T_k-s)/\varepsilon_i} = -\frac{\varepsilon_i}{T_k - s} \, . 
\end{equation*}
which grows to negative infinity since $\varepsilon_i$ converges to a positive limit. Hence, $\eta^{k,i}$ is not square-integrable. 
}, invalidating a proof that relies on the integrability of the relative density volatility. We therefore turn towards investigating the relationship between (a positive multiple of) the calendar spread $R^{j, i} = \Sigma^{j, i} q^{j, i}$ and (a positive multiple of) the butterfly spread $q^{j, i}$ near settlement, which in turn will allow us to read off the DLV. 

Throughout this section fix the front expiration $k$ with expiration time $T_k$ and denote by $J$ the node at which spot settles: $X_{T_k} = K^J$. To prove the existence of $L$ on $[T_{k-1}, T_k]$ we make the following two assumptions: 
\begin{assumption}
\phantomsection\label{assumption:existence_mw}
\begin{enumerate}[label=(\roman*)]
  \item For the front, the density has mass on the interior strikes\footnote{Note that this is a non-restrictive assumption: for finite $L_t^{j}$, $M^j$ is a non-singular M-matrix whose inverse has strictly positive interior columns. This means that if spot settles at $T_{k-1}$ at an interior node $K^J$, we have positive mass $q^k_{T_{k-1}}=(M^k)^{-1}e_J$. The density has zero interior mass only if the spot got absorbed by an exterior node $e_0$ or $e_{N+1}$. Here $e_J$ denotes the $J$-th unit vector.}
  $$\sum_{i=1}^N q_{T_{k-1}}^{k, i} > 0 \, .$$
  \item For all $j=k+1, \dots, M$ and $i=1, \dots, N$, the vol-of-vol $\xi^{j, i}$ is locally Lipschitz continuous in $(t, B, L)$ and is uniformly bounded in $(B, L)$ by a square-integrable function in time: 
  \begin{equation*}
    \|\xi^{j, i}(t, B, L)\| \leq c_t \quad \textrm{where} \quad \int_{T_{k-1}}^{T_k}c_t^2dt<\infty
  \end{equation*}
  where $\|\cdot\|$ is the Euclidean norm.
\end{enumerate}
\end{assumption}
%Note that assumption (iii) is consistent with the constant loadings that we estimate for our GROMM in Section~\ref{sec:estimation}.

\begin{theorem}[Existence]
\label{thm:existence_bh}
Let $L_{T_{k-1}} $ be finite. Under Assumption~\ref{assumption:existence_mw} the following holds true: 
\begin{enumerate}[label=(\roman*)]
  \item There is a unique strong solution for $L$ on $[T_{k-1},T_k)$, and every $\Sigma^{j,i}$ stays
    finite and strictly positive.
  \item If spot settles at an interior strike, $ J \in \lbrace 1, \dots, N \rbrace $, then $L$ extends continuously to $T_k$
    with $L_{T_k}$ finite.
  \item If spot settles at an extreme strike $ J \in \lbrace 0, N +1 \rbrace$, then $q^j_{T_k} = e_J$ for all $j > k$. Moreover, the remaining DLVs $\Sigma$ no longer affect the model. 
\end{enumerate}
\end{theorem}

Before we prove the existence statement, we demonstrate that two positive local martingales that only differ by a square-integrable relative volatility vanish together: 
\begin{lemma}
\label{lem:zero_sets_bh}  
Let $\tau\le T_k$ be a stopping time and let $U, V$ be two positive continuous local martingales defined on the interval $[T_{k-1},\tau)$. Assume that $U, V$ have relative volatilities $u$ and $u + v$ respectively, where $u, v$ take values in $\R^d$, i.e. $dU / U = u \,dW$ and $dV / V = (u+v) \, dW$. Then the following holds true: 
\begin{enumerate}[label=(\alph*)]
  \item The left limit $U_{\tau-}$ exists, and $U_{\tau-}$ is zero exactly when $\int_{T_{k-1}}^{\tau} \| u_t \|^2 dt = \infty $.
  \item If $v$ is square-integrable on the interval $[T_{k-1}, \tau)$, then $U_{\tau-}=0$ if and only if $V_{\tau-}=0$.
\end{enumerate}
\end{lemma}
\begin{proof}
  The idea of the proof is simple: $U$ is an exponential, and goes to zero if its quadratic variation diverges. 
  \begin{enumerate}[label=(\alph*)]
    \item Write $U_t$ as an exponential local martingale: $U_t = U_{T_{k-1}}\exp(\mathcal{M}_t - \frac{1}{2}\langle \mathcal{M} \rangle_t)$ where $\mathcal{M}_t = \int_{T_{k-1}}^t u_s \, dW_s$ and $\langle \mathcal{M} \rangle_t = \int_{T_{k-1}}^t\|u_s\|^2ds$ denotes the quadratic variation up to time $t$. If $\langle \mathcal{M} \rangle_\tau$ is finite, then $\mathcal{M}$ converges to a finite limit and hence $U_{\tau-}>0$. Otherwise, if $\langle \mathcal{M} \rangle_\tau=\infty$, then by Dambis--Dubins--Schwarz we identify the local martingale $\mathcal{M}_t$ with some Brownian motion $ \beta_{\langle \mathcal{M} \rangle_t} $ run on the clock $s = \langle \mathcal{M} \rangle_t$. For $\beta$, the unit increments $\beta_{n+1} - \beta_{n}$ are i.i.d. $N(0,1)$. Therefore, by the strong law of large numbers, $\beta_s / s \to 0$ almost surely. Thus, $\beta_s$ has sub-linear growth and we obtain that $\mathcal{M}_t - \frac{1}{2}\langle \mathcal{M}\rangle_t = s( \beta_s/s - \frac{1}{2}) \to -\infty $ and $U_t \to U_{\tau-} = 0$.
    \item The same holds true for $V$, which has volatility $u+v$. Since $\int_{T_{k-1}}^\tau\|v_t\|^2 \, dt<\infty$, the integrated sum $\int_{T_{k-1}}^\tau\|u_t+v_t\|^2 \, dt$ is finite if and only if $\int_{T_{k-1}}^\tau\|u_t\|^2 \, dt$ is, so by (a) $U_{\tau-}$ and $V_{\tau-}$ are either both zero or both positive.
  \end{enumerate}
\end{proof}

\begin{proof}[Proof of Theorem~\ref{thm:existence_bh}]
\begin{enumerate}[label=(\roman*)]
  \item Since $\Sigma^{j, i} = \exp(L^{j, i}) > 0$ for all $j, i$, every $ M^j $ is a non-singular M-matrix with an inverse that has strictly positive interior columns. By Assumption~\ref{assumption:existence_mw}~(i) and the front construction, the front has interior mass for $t < T_k$. Thus $q^{k+1}=(M^{k+1})^{-1}q^k > 0$, and inductively $q^j>0$ for all $j>k$. Moreover, the coefficients of \eqref{eq:log_dlv_sde_mw} are locally Lipschitz. Hence, there is a unique strong solution up to a maximal lifetime $\zeta \leq T_k$ \citep[Theorem~V.38]{Protter2005}. It remains to show $\zeta = T_k$.
  
  Let $R^{j, i} \coloneqq \Sigma^{j, i} q^{j, i}$ and recall that the density $q^{j}$ is governed by $dq^{j, i} = q^{j, i} \eta^{j, i} \, dW$. We want to show that both processes fulfil the assumptions of Lemma~\ref{lem:zero_sets_bh}. Applying the It\^o product rule to $R^{j, i}$ we observe that it is a local martingale $dR^{j,i} = R^{j,i} ( \eta^{j,i}+\xi^{j,i} ) \, dW$. Moreover, we can show that $R^j$ is bounded by applying the call payoff matrix $\Pi$, where $\Pi_{l, m} = (K^m - K^l)^+$, to the code book relation \eqref{eq:hb_dlv} expressed in terms of $R^j$
  \begin{equation}
    \label{eq:R_codebook_mw}
    q^{j-1}=q^j - \Omega_jR^j 
  \end{equation}
  which gives $C^{j} - C^{j-1} = \Pi\Omega_j R^j $ with $C^j = \Pi q^j$ the call prices. Since $\Pi\Omega_j$ is diagonal on the interior nodes, $R^{j, i}$ is the calendar spread at $K^{i}$ from $T_{j-1}$ to $T_j$ up to a positive constant. Since $X \in [K^0, K^{N+1}]$, call prices lie in $[0, K^{N+1}]$, so $R^{j}$ must be bounded. Thus, both $q^j$ and $R^j$ are bounded martingales and therefore attain limits at $\zeta$ by the martingale convergence theorem. Apply Lemma~\ref{lem:zero_sets_bh} (b) with $\tau=\zeta$ and $v=\xi$, then both $q^{j, i}$ and $R^{j, i}$ can only reach zero together. 

  Next we show that a single zero-valued interior density node forces all interior nodes to be zero-valued. In contrast, a single positive-valued interior node spreads to all interior nodes. Abbreviate $p=q^{j-1}$, $q=q^j$, $R=R^j$ at time $\zeta$ and write out row $i$ of \eqref{eq:R_codebook_mw}: 
  \begin{equation*}
    p_i = q_i+( \gamma^{-i} + \gamma^{+i} ) R_i - \gamma^{ +(i-1)} R_{i-1} - \gamma^{ -(i+1)} R_{i+1}.
  \end{equation*}
  From Lemma~\ref{lem:zero_sets_bh} (b) we know that $q_i = 0$ iff $R_i = 0$. Hence, if $q_i = 0$, then also $R_i = 0$ and $p_i = - \gamma^{ +(i-1)} R_{i-1} - \gamma^{ -(i+1)} R_{i+1}$. The left-hand side, $p_i$, is non-negative and the right-hand side non-positive. Therefore $p_i = 0$ and $R_{i\pm1} = 0$ and again by Lemma~\ref{lem:zero_sets_bh} (b) $q_{i\pm1} = 0$. This argument spreads across the grid such that all interior nodes have to be equal to zero. Mass can only be attained on an outer node. Conversely, if $q^{j-1}$ has interior mass then the chain above cannot start and all $q^{j, i}$, and thus $R^{j, i}$, are positive. 

  Now, suppose that $\zeta < T_k$. Then the Brownian motion $B_\zeta$ is finite, and the front density $q^k_\zeta$ still has interior mass. Applying the above inductive chain for $j = k+1, \dots, M$ we obtain that $\Sigma^{j, i}_{\zeta-} = R^{j, i}_{\zeta-} / q_{\zeta-}^{j, i} \in (0, \infty)$. Hence, $L$ has a finite limit at $\zeta$ which means it could be extended to $\zeta$, contradicting the maximality of $\zeta$. Therefore $\zeta = T_k$, which proves existence.

  \item If $J \in \lbrace 1, \dots, N \rbrace $, the terminal front has interior mass and by the above inductive chain $q^{j,i}_{T_k},R^{j,i}_{T_k}>0$ for all $j > k$. Hence, $\Sigma^{j, i}_{T_k} = R^{j, i}_{T_k} / q^{j, i}_{T_k} \in (0, \infty)$ for all $j > k$, and therefore the solution of $L$ extends to $L_{T_k}$. 
  
  \item Since $K^\top \Omega_j = 0$, the code book conserves the mean. Hence, for all $j > k$ we have $K^\top q^j_{T_k} = K^\top q^k_{T_k} = X_{T_k} = K^J$. A probability vector $q^{j}_{T_k}$ defined on $K^0, \dots, K^{N+1}$ with mean on the extreme end forces $q^j_{T_k} = e_J$, $J \in \lbrace 0, N+1 \rbrace$. Again recalling the code book recursion $q^{j-1} = (I - \Omega_j S^j) q^j$ and that $S^j = \diag(0, \Sigma^{j, 1}, \dots, \Sigma^{j, N}, 0)$ it follows that for an exterior density $q^j = e_J$, we have $S^j e_J = 0$. Thus, $e_J = M^j e_J$ for all $\Sigma$ and by induction $ q^{j}_{T_k} = (M^j)^{-1} q^{j-1}_{T_k} = (M^j)^{-1} e_J = e_J$ for all $j > k$. Hence all future densities, and therefore also all call option prices, are independent of $\Sigma$. 
\end{enumerate}
\end{proof}

Whenever the spot settles on an interior strike, the next front shifts $k \rightarrow k +1 $ with density $(M^{k+1}_{T_k})^{-1}e_J$. The theorem can then be re-applied to the interval $[T_k, T_{k+1}]$. If the spot settles at an extreme strike, (iii) applies and the spot process is absorbed at this state. Since there are finitely many expirations, the existence theorem therefore gives the process on the full interval $[0, T_M]$. 

\begin{corollary}
\label{cor:true_martingale_mw}
Under the assumptions of Theorem~\ref{thm:existence_bh}, $q^j$ is a true martingale on $[T_{k-1},T_k]$ for every $j\geq k$.
\end{corollary}
\begin{proof}
By construction the drift of $q^j$ vanishes identically (for $j>k$ by the node-by-node cancellation in \eqref{eq:drift_hb}, for $j=k$ by the Gaussian bridge construction), so $q^j$ is a local martingale on the interval of existence furnished by Theorem~\ref{thm:existence_bh}. Each coordinate $q_t^{j,i}$ is a probability and therefore satisfies $0\leq q_t^{j,i}\leq 1$, i.e.\ $q^j$ is a bounded process. A bounded local martingale is uniformly integrable and closed by its terminal value, hence a true martingale; in particular $\mathbb E_t[q_s^j]=q_t^j$ for all $T_{k-1}\leq t\leq s$ in the interval, with no additional integrability condition required.
\end{proof}

\subsection{Sampling}
\label{sec:sampling}

  In this section we discuss the sampling procedure for the discrete local volatility model. The final model is a product of this 
  process and a simple Black~\&~Scholes process with, say, volatility of 5\%. It serves as a SANOS~\citep{SANOS2026} smoother and ensures
  that the overall option price surface is smooth in strike.

   On a time grid
  $t_n<t_{n+1}$ with $\delta_n=t_{n+1}-t_n$, one may proceed as follows.

  First draw one Brownian vector $\Delta W_n\in\mathbb R^d$, using its first
  coordinate both in \eqref{eq:front-exact-step} and in every DLV coordinate
  with a non-zero loading $\xi^{j,i;1}$. This preserves the specified leverage.
  Update the front density by \eqref{eq:front-exact-step}, and calculate spot
  from \eqref{eq:spot-exact-step}. This part of the update has no time-step
  bias, preserves positivity and unit mass, and is exactly martingale
  preserving.

  For the DLV squared it is useful to work with
  $L_t^{j,i}\coloneqq\log\Sigma_t^{j,i}$, which follows \eqref{eq:log_dlv_sde_mw}. 

  Even the frozen-coefficient exponential update
  \begin{equation}
    \label{eq:dlv-log-euler}
    \Sigma_{n+1}^{j,i}
      =\Sigma_n^{j,i}\exp\!\left\{
        \left(\mu_n^{j,i}
          -\frac12\sum_f(\xi_n^{j,i;f})^2\right)\delta_n
        +\sum_f\xi_n^{j,i;f}\Delta W_n^f
      \right\}
  \end{equation}
  is preferable to Euler integration of $\Sigma$: it cannot create a negative
  local variance. For state-dependent coefficients, a genuine It\^o
  weak-order-two stochastic Runge--Kutta method in the log variables reduces
  the remaining time-step bias.

  After updating $\Sigma^{j}$ reconstruct the densities expiry by expiry by
  solving the tridiagonal systems
  \begin{equation}
    \label{eq:density-algebraic-step}
    q_{n+1}^{j}
      =\left(I-\Omega_jS_{n+1}^{j}\right)^{-1}q_{n+1}^{j-1},
    \qquad j=k+1,\ldots,M.
  \end{equation}
  This enforces the DLV relation at every simulated state. Because the matrix
  is an M-matrix for non-negative $\Sigma$, the solve preserves positivity.

  The drift condition \eqref{eq:drift_hb} removes the infinitesimal drift of
  every density coordinate. It does not, by itself, make a generic finite-step
  discretization exactly unbiased, since the map
  $\Sigma^j\mapsto(M^j)^{-1}q^{j-1}$ is nonlinear. Thus
  \eqref{eq:dlv-log-euler}, or any ordinary higher-order SDE scheme, retains a
  time-discretization bias which should vanish under step refinement.

  Our experiments in Section~\ref{sec:numerics} use this sampling scheme with antithetically paired Brownian increments.

\begin{remark}[Alternative discrete-branch construction]
  \label{rem:discrete-branch-sampling}
  If exact martingality under a chosen finite-step simulation law is required,
  it can instead be imposed directly. Choose symmetric discrete shock branches
  $\Delta W_{n,b}$ with probabilities $w_b$ and construct the front-density
  branches so that
  \[
    \sum_b w_bq_{n+1,b}^k=q_n^k \ .
  \]
  Parametrize each positive DLV update with an unknown drift correction
  $c_n^j\in\mathbb R^N$. For example, replace the log drift in
  \eqref{eq:dlv-log-euler} by $c_n^{j,i}$. For each expiry in turn, solve the
  nonlinear system
  \begin{equation}
    \label{eq:finite-step-drift-correction}
    \sum_b w_b q_{n+1,b}^j(c_n^j)=q_n^j,
    \qquad
    q_{n+1,b}^j(c_n^j)
      =\left(I-\Omega_jS_{n+1,b}^j(c_n^j)\right)^{-1}
        q_{n+1,b}^{j-1}.
  \end{equation}
  Only $N$ coordinates are independent because mass and first moment are
  already preserved, matching the $N$ drift corrections. Initializing the
  nonlinear solver with the continuous-time drift
  \eqref{eq:drift_hb} gives
  a stable finite-step correction. Equation
  \eqref{eq:finite-step-drift-correction} makes the densities exactly
  martingales under the discrete branch law; symmetric or antithetic branches
  reduce variance but do not replace this condition.

\end{remark}

\section{Estimation}\label{sec:estimation}

We briefly describe our estimation strategy. There are two components: fitting daily the ``implied'' discrete densities and therefore discrete local volatilities
to observed bid-ask spreads, and then an estimation of the vol of those discrete local volatilities, and their statistical drift.

The former represents today's fit of the model to the market, while the latter captures the evolution of the model parameters over time.

The summary here is brief; refer
to DYSANOS~\citep{DYSANOS2026} for a comprehensive discussion on parameter mapping and estimation considerations.

\subsection{Data and state representation}
\label{sec:estimation-data}

We estimate a family of models which has 25 expiries at one-business-day intervals. The model is aimed at S\&P~500 index options
whose main trading activity occurs in the short-dated options segment.

Each daily state has 25 model expiries at one-business-day intervals and 200
interior strike nodes, together with two artificial boundary nodes. The 200
inner nodes are generated from fixed standardized log-moneyness templates on
$[-2.5,1.5]$.\footnote{The expiry grid is fixed at $T_j=j/252$,
$j=1,\ldots,25$.  A normalized-moneyness coordinate $z$ at expiry $\tau$ is
mapped to the forward-relative strike
\[
  k(z,\tau)=\exp\!\left[z\sqrt{
    \frac{\sigma_{\epsilon}^2\epsilon+
    \sigma_{\mathrm{ATM}}(\tau)^2\tau}{1+\epsilon}}
  \right],
  \qquad \epsilon=\frac{1}{8\times252},
\]
where $\sigma_{\epsilon}=\sigma_{\mathrm{ATM}}(1/252)$ and
$\sigma_{\mathrm{ATM}}(\tau)$ is the regularized annual ATM Black volatility.
The cash strike is $K=F_0(\tau)k$; in the normalized simulations $F_0(\tau)=1$.
The strike construction uses 100 fixed coordinates for the one-day ATM scale
and 100 for the 25-day ATM scale, both on $[-2.5,1.5]$.  On each calibration
date the resulting log-relative-strike candidates are merged and sorted, and
given a minimum gap equal to 5\% of the average log-strike spacing.
Exponentiation gives one 200-node relative-strike grid shared by all 25
expiries; the boundary strikes are one-half of the lowest and one and a half
times the highest inner strike.  Thus the expiry and normalized-moneyness
templates are fixed, and the resulting cash strikes are fixed along each
simulated path.  For historical calibration the relative-strike grid floats
between dates because the ATM scales change.}
They therefore share one strike support across the 25 expiries on a given
date, while that support floats between dates with the observed ATM term
structure. A fixed 5\% annual Black volatility supplies the independent SANOS
smoothing factor.

We use European S\&P~500 index-option quotes via WRDS\footnote{\url{https://wrds-www.wharton.upenn.edu/}} from OptionMetrics IvyDB
on 891 valid market dates from 3 January 2022 to 29 August 2025.
After the cleaning and eligibility pipeline described in
Appendix~\ref{app:market-pipeline}, the retained sample contains 2,192,458
original call and put bid-ask intervals with at most 25 business days to
expiry. Prices and strikes are forward-normalized, while the contemporaneous
spot is retained separately.

\subsection{End-to-end fitting}
\label{sec:estimation-fit}

Let $x_d\in\mathbb R^{25\times200}$ denote the unconstrained state on date
$d$. Its DLV-squared coordinates are decoded as
\[
  \Sigma_d^{j,i}=\min\{9,10^{-4}+\exp(x_d^{j,i})\}.
\]
The SANOS recursion then produces all densities and quote prices. Every
retained market quote is priced at its own strike and expiry; quotes between
two model expiries use the corresponding linear time interpolation. The fit is
therefore performed directly against the original market bid and ask, not
against an intermediate linear-programming state.

The \emph{direct} specification fits all 5,000 coordinates of $x_d$
independently for every date. The \emph{affine} specifications impose
the shared representation
\begin{equation}
  \label{eq:estimation-affine-state}
  x_d=\bar x+\sum_{a=1}^r z_d^a u_a,
  \qquad r\in\{100,200,400\},
\end{equation}
and jointly estimate the mean $\bar x$, loadings $u_a$, and daily codes $z_d$. See~\cite{DYSANOS2026}
for motivation and further results.

For model price $C$, quoted interval $[b,a]$, midpoint $m=(a+b)/2$, and
half-spread $h=(a-b)/2$, define the bid-ask \emph{excess}
\begin{equation}\label{eq:excess}
  e(C;b,a)=\frac{\max(C-a,0)+\max(b-C,0)}{h}.
\end{equation}
Writing $n_d$ for the number of quotes on date $d$, the quote component of the
training objective is
\begin{equation}
  \label{eq:estimation-quote-objective}
  \mathcal L
  =\frac1{891}\sum_d\frac1{n_d}\sum_{\ell=1}^{n_d}
    \left[
      e(C_{d\ell};b_{d\ell},a_{d\ell})
      +0.05\frac{\abs{C_{d\ell}-m_{d\ell}}}{h_{d\ell}}
    \right],
\end{equation}
plus a soft barrier near the upper DLV-squared bound, cf.~\cite{DYSANOS2026}. Thus a price within its
quoted interval has zero excess, while the small midpoint term identifies a
stable point inside the spread. We use the exact all-quote gradient, accumulated
in memory-bounded quote blocks, and Adam for up to 1,000 epochs.

Table~\ref{tab:numerics-static-market-fit} reports statistics pooled over all
2,192,458 retained quotes. Conditional excess is the mean of $e$ only over
quotes outside their spread; midpoint MAE is also measured in half-spreads.

\begin{table}[H]
\centering
\small
\begin{tabular}{lrrrrr}
\toprule
Representation & State dim. & Inside bid-ask & Mean excess
  & Conditional excess & Midpoint MAE\\
\midrule
Direct 25-by-200 & 5,000 & 97.32\% & 0.035 & 1.288 & 0.237\\
Affine-100       &   100 & 89.68\% & 0.083 & 0.808 & 0.471\\
Affine-200       &   200 & 95.02\% & 0.045 & 0.900 & 0.349\\
Affine-400       &   400 & 96.01\% & 0.040 & 1.009 & 0.317\\
\bottomrule
\end{tabular}
\caption{End-to-end fit to the retained market bid-ask intervals. State
dimension is 5,000 for the unrestricted daily state and the latent rank for an
affine state.}
\label{tab:numerics-static-market-fit}
\end{table}

The direct representation gives the highest bid-ask coverage and the lowest
midpoint error and is therefore used for the stochastic-DLV experiments below.
Affine-200 already exceeds 95\% bid-ask coverage; increasing the rank to 400
gives a smaller further improvement. Conditional excess is computed only on
the changing set of failed quotes and should not be read as a monotone ranking
across representations.

\subsection{DLV vol-of-vol estimation}
\label{sec:estimation-vol-of-vol}

The floating historical supports cannot be differenced node by node without
first aligning contracts.  We use 710 adjacent market-session pairs ending no
later than 31 December 2024.  Write
\[
  L_d^j(x):=\log\Sigma_d^j(x),
  \qquad R_d:=\frac{S_{d+1}}{S_d},
\]
and let $x_{d,i}$ be the $i$th relative strike on the origin-date support.  In
the zero-carry normalization used for the dynamics, holding the corresponding
cash strike fixed moves its relative strike to $x_{d,i}/R_d$ on the next date.
Holding calendar expiry fixed simultaneously shifts the remaining-expiry index
down by one.  The aligned observation is therefore
\begin{equation}
  \label{eq:estimation-fixed-contract-increment}
  \Delta L_d^{j,i}
  =L_{d+1}^{j}\!\left(\frac{x_{d,i}}{R_d}\right)
   -L_d^{j+1}(x_{d,i}),
  \qquad j=1,\ldots,24,\quad i=1,\ldots,200.
\end{equation}
The endpoint value is interpolated on date $d+1$'s own floating support.  Thus
cash strike and calendar expiry are fixed within each one-day pair, while the
origin chart is allowed to float across different pairs.

For each expiry--strike coordinate we fit the physical one-step regression
\begin{equation}
  \label{eq:estimation-dlv-ar}
  L_{d+1}^{j}\!\left(\frac{x_{d,i}}{R_d}\right)
  =a^{j,i}+b^{j,i}L_d^{j+1}(x_{d,i})+\varepsilon_d^{j,i}.
\end{equation}
The centred residual surfaces are smoothed by a Gaussian filter with widths
one expiry node and two strike nodes.  Let
\[
  Z_d=\Phi^{-1}\!\left(F_d^1(R_d)\right),
  \qquad
  \widetilde Z_d=\frac{Z_d-\overline Z}{s_Z},
\]
where $F_d^1$ is the fitted front distribution including the independent 5\%
Black factor.  The daily loading of the spot factor is
\[
  \ell_1^{j,i}=\frac1n\sum_d
    \widetilde\varepsilon_d^{j,i}\widetilde Z_d.
\]
We remove this component and apply PCA to
$\widetilde\varepsilon_d-\widetilde Z_d\ell_1$; its leading eigenvectors give
the remaining daily loading surfaces $\ell_2,\ell_3,\ldots$.  The annualized
DLV vol-of-vol used in the SDE is
\begin{equation}
  \label{eq:estimation-annualized-xi}
  \xi^{j,i;f}=\sqrt{252}\,\ell_f^{j,i}.
\end{equation}
The fitted spot scores have mean $0.0290$ and standard deviation $1.0026$.
Eight factors retain 72.88\% of the smoothed innovation variance and 16 retain
84.40\%; 26 factors are required to exceed 90\%, reaching 90.25\%.  The
rank-16 system is nested, with its first eight loading surfaces identical to
the rank-8 system.  Appendix~\ref{app:dlv-factor-figures} shows the baseline
implied-volatility surface, the cumulative factor power, and the implied-volatility
response to each retained factor.  On the held-out 2025 sample the
full 25-by-200 representation has residual root mean square 0.256, compared
with 0.423 for the Affine-200 representation.

This is a functional estimate in remaining-expiry and normalized-moneyness
coordinates, not an estimate on one global historical cash-strike grid.  When
a simulation is initialized, these loading surfaces are attached nodewise to
that state's fixed cash-strike grid.  Pooling dates therefore assumes that the
loading functions are stable in this moving coordinate chart.

The empirical physical log-DLV-squared drift implied by
regression~\eqref{eq:estimation-dlv-ar} is
\begin{equation}
  \label{eq:estimation-empirical-p-drift}
  \widehat\mu_{d,k}^{\PP,\mathrm{emp}}
  =252\left[a_k+(b_k-1)L_{d,k}\right],
\end{equation}
where $k$ denotes an aligned expiry--strike coordinate.  At the same
historical state we compute $\mu_d^\Q$ as the risk-neutral log drift
in~\eqref{eq:log_dlv_sde_mw}.  Under the convention
$\dd W^\Q=\dd W^\PP+\lambda\,\dd t$, a physical drift compatible with the
fitted diffusion must have the form
\begin{equation}
  \label{eq:numerics-p-drift}
  \mu_d^{\PP,\mathrm{comp}}=\mu_d^\Q+\xi\lambda_d.
\end{equation}
We estimate the state-dependent market price of risk by
\begin{equation}
  \label{eq:numerics-statewise-lambda}
  \widehat\lambda_d
  =\arg\min_{\lambda\in\R^d}
    \sum_k w_{d,k}
    \left[
      \widehat\mu_{d,k}^{\PP,\mathrm{emp}}-\mu_{d,k}^{\Q}
      -(\xi\lambda)_k
    \right]^2.
\end{equation}
Here $w_{d,k}$ is proportional to the interior risk-neutral density, normalized
so that each expiry has equal total weight.  We also estimate one constant
$\widehat{\bar\lambda}$ by pooling the same objective over all dates; this is
the measure change used in the simulations.  The independent 5\% Black factor
is assigned zero risk premium and therefore has the same law under $\PP$ and
$\Q$.

For comparison with a model-free time-series estimate, define the
unconditional coordinate drift
\begin{equation}
  \label{eq:estimation-unconditional-p-drift}
  \widehat{\bar\mu}^{\PP,\mathrm{unc}}_k
  =\frac{252}{n}\sum_{d=1}^n\Delta L_{d,k}.
\end{equation}
Because~\eqref{eq:estimation-dlv-ar} contains an intercept, this is also the
sample mean of~\eqref{eq:estimation-empirical-p-drift}.  We measure the distance
of a compatible drift from this unconditional estimate by
\begin{equation}
  \label{eq:estimation-unconditional-drift-distance}
  D
  =\left[
    \frac{\sum_k\bar w_k
      \left(\overline{\mu^{\PP,\mathrm{comp}}}_k
      -\widehat{\bar\mu}^{\PP,\mathrm{unc}}_k\right)^2}
    {\sum_k\bar w_k}
  \right]^{1/2},
  \qquad
  \bar w_k=\frac1n\sum_{d=1}^n w_{d,k}.
\end{equation}

\begin{table}[H]
\centering
\small
\begin{tabular}{lrr}
\toprule
Calibration quantity & Rank 8 & Rank 16\\
\midrule
PCA innovation variance retained & 72.88\% & 84.40\%\\
$\Q$-drift RMS & 13.58 & 13.77\\
Empirical $\PP-\Q$ target RMS
  & 39.41 $(38.95,39.88)$ & 39.29 $(38.82,39.75)$\\
Statewise compatible residual RMS
  & 24.23 $(24.01,24.44)$ & 20.81 $(20.63,20.99)$\\
Statewise target energy captured
  & 62.20\% $(61.64,62.75)$ & 71.95\% $(71.52,72.37)$\\
Constant-$\lambda$ residual RMS
  & 29.97 $(29.58,30.37)$ & 28.54 $(28.14,28.94)$\\
Constant-$\lambda$ target energy captured
  & 42.15\% $(40.94,43.32)$ & 47.24\% $(46.10,48.37)$\\
$D$ to unconditional drift, statewise $\lambda_d$
  & 16.33 $(15.75,16.90)$ & 13.57 $(13.14,14.01)$\\
$D$ to unconditional drift, constant $\bar\lambda$
  & 16.43 $(15.80,17.05)$ & 13.73 $(13.21,14.26)$\\
\bottomrule
\end{tabular}
\caption{Vol-of-vol coverage and physical-drift projection on the 710
one-session training pairs.  Parentheses give 95\% intervals.  RMS quantities
are annualized log-DLV-squared drifts.  The target and residual intervals use
20,000 pair-bootstrap resamples; the intervals for $D$ use the pair-level
delta method.  They condition on the fitted AR, PCA, and risk-neutral drift,
so they quantify time-series sampling error rather than full re-estimation
uncertainty.}
\label{tab:numerics-calibration}
\end{table}

The statewise fit is a diagnostic upper bound on what the retained diffusion
span can reproduce.  The constant-$\lambda$ rows describe the more restrictive
measure change actually used below.  Rank 16 materially improves both fits,
but neither retained factor system reproduces the empirical physical drift
exactly.  In particular, the confidence intervals in the last two rows do not
include zero; the compatible drift remains measurably different from the raw
unconditional estimate.

Appendix~\ref{app:drift} visualizes the effect of the drift change
with four matched Model~D paths (Section~\ref{sec:numerics}).  

\section{BHW No Dynamic Arbitrage Test}\label{sec:numerics}

In this section we assess whether the model is indeed free of dynamic arbitrage.
We are not aware of an accessible numerical benchmark for dynamic arbitrage in the literature, so
the methodology described here is a contribution in its own right, first presented in~\cite{DYSANOS2026}.

We compare four model specifications:
\begin{center}
\begin{tabular}{crr}
\toprule
Model & DLV factors & Steps/day\\
\midrule
A & 8  & 1\\
B & 16 & 1\\
C & 8  & 8\\
D & 16 & 8\\
\bottomrule
\end{tabular}
\end{center}
Every statistical audit uses 100,000 paths arranged as 50,000 independent antithetic
pairs.  The restricted-universe audits use a 20-business-day horizon; the full
instrument-grid extension below uses 25 business days.  For a live instrument $i$, the
one-day return used in both diagnostics is
\begin{equation}\label{eq:numerics-spot-scaled-return}
  R_{p,n}^i
  =\frac{V_{p,n+1}^i-V_{p,n}^i}{Z_{p,n}},
\end{equation}
where $Z$ is spot and $V^i$ is either spot or an option price.  Division by the
pre-step spot is predictable and therefore preserves the zero-conditional-mean
property under~$\Q$.

In addition to GROMM we also ran the two reference models from the literature
\begin{itemize}
  \item \textbf{Tangent L\'evy}~\citep{TangentLevy}: a finite-grid jump code book
  whose conditional L\'evy density obeys the convolutional HJM drift restriction,
  with spot driven by the compensated short-end jump law. This model is theoretically
  free of dynamic arbitrage.

  We use 16 historical
  PCA factors and eight simulation steps per day; Appendix~\ref{app:tangent-levy}
  records the calibration, reconstruction, and positivity choices.

  \item \textbf{CRC}~\citep{DysLitGambaraTeichmann2021}: a generalized-Bates
  consistent-recalibration model with Heston variance and five
  maturity-bucketed compensated Gaussian jump laws.  We estimate the reset
  law strictly before the reference date and use the inverse recalibration
  map described in Appendix~\ref{app:bates-crc}. This model class would be free of arbitrage
  if the calibrated model satisfied Definition~6.3, which for our setup is
  not the case. This reflects the challenge that the model relies on solving
  ``an inverse problem that is ill-posed in the sense of Hadamard'', page~26.
\end{itemize}

\subsection{Initial Fit}

We compare all three initial states on the same retained quote book from
29 August 2025: 2,662 call and put intervals with at most 25 business days to
expiry and the GROMM normalized-moneyness range.  Figure~\ref{fig:initial-fit-iv}
shows signed Black implied-volatility excess beyond the bid-ask volatility
interval, and Figure~\ref{fig:initial-fit-spread} shows the analogous signed
price excess~\eqref{eq:excess} in half-spread units.  The figures plot only
quotes with actual excess beyond a boundary; quotes fitted inside the spread
are omitted.  Each panel title reports its percentage outside the spread.  The
price plot uses a symmetric-log colour
scale, linear between $-1$ and $1$, so small excess remains visible despite
the much larger Tangent L\'evy errors.

\begin{figure}[H]
  \centering
  \includegraphics[width=\textwidth]{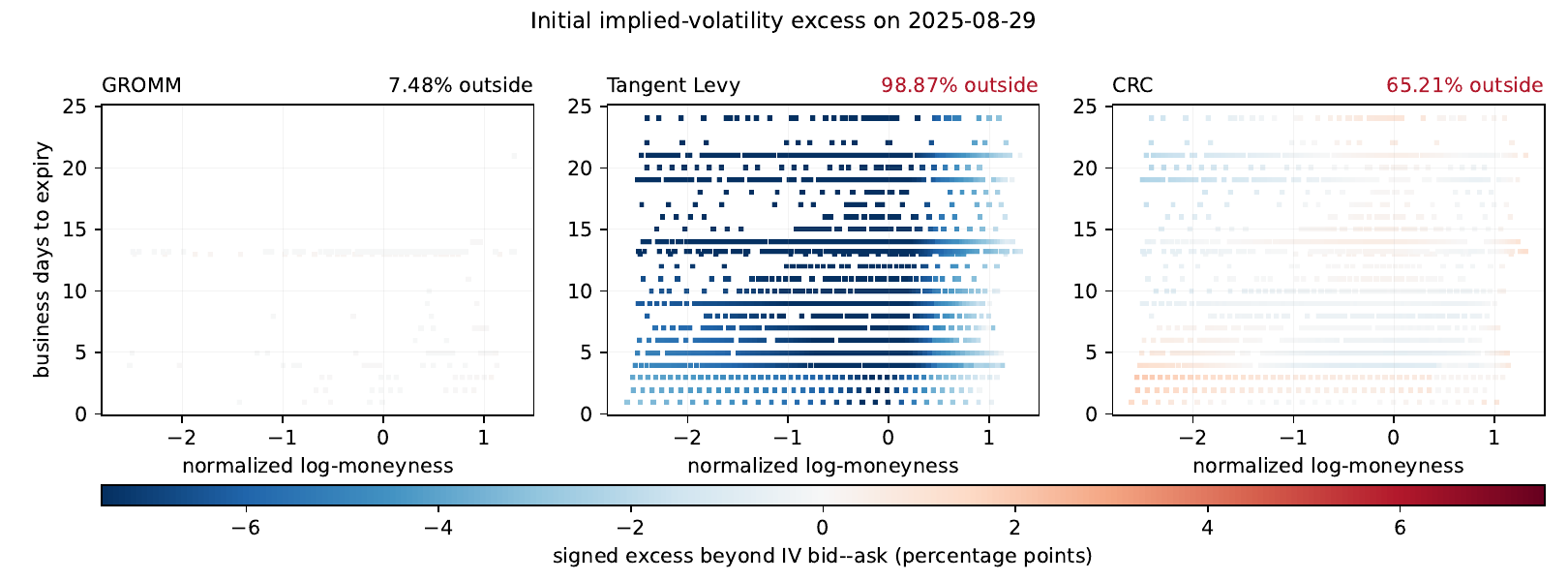}
  \caption{Initial implied-volatility excess across the common 29 August 2025
  surface.  Colours show signed distance beyond the nearest bid-ask boundary
  in Black-volatility percentage points.  Quotes inside the spread are not
  plotted; all three panels use the same scale.  Red percentages identify
  poor static fits; they are not dynamic no-arbitrage tests.}
  \label{fig:initial-fit-iv}
\end{figure}

\begin{figure}[H]
  \centering
  \includegraphics[width=\textwidth]{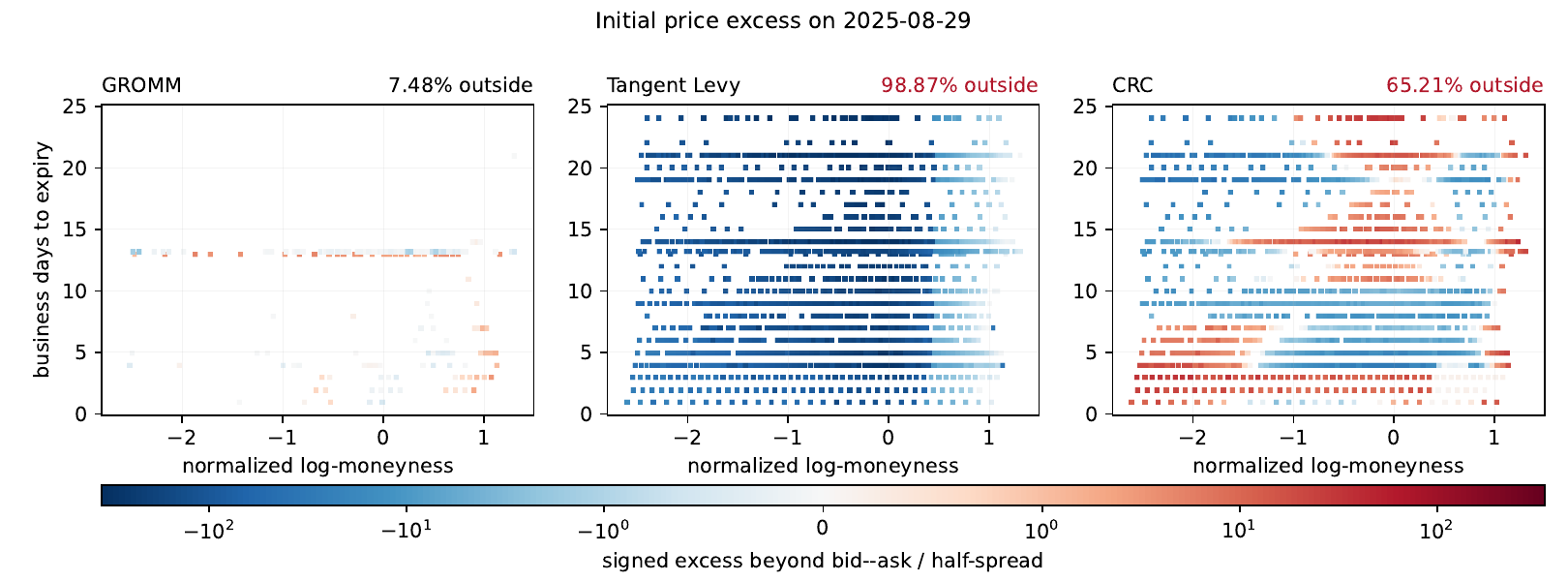}
  \caption{Initial signed price excess across the same surface, measured from
  the nearest bid-ask boundary in half-spread units.  Quotes inside the spread
  are not plotted; the common symmetric-log colour scale has its linear region
  at one half-spread.  Red percentages identify poor static fits only.}
  \label{fig:initial-fit-spread}
\end{figure}

\begin{table}[H]
\centering
\small
\begin{tabular}{lrrr}
\toprule
Model & Quotes with excess & Conditional mean & Conditional median\\
\midrule
GROMM & 7.48\% & 0.971 & 0.358\\
Tangent L\'evy & \textcolor{red}{98.87\%} & 112.486 & 102.046\\
CRC & \textcolor{red}{65.21\%} & 9.739 & 7.129\\
\bottomrule
\end{tabular}
\caption{Initial-fit excess on the common 29 August 2025 quote book.
Quotes with excess is the percentage with $e>0$~\eqref{eq:excess}; the final two
columns report the mean and median of $e$ conditional on $e>0$.
The GROMM row uses the same checkpoint as
Table~\ref{tab:numerics-static-market-fit}, restricted to this date.}
\label{tab:initial-fit-excess}
\end{table}

Tangent L\'evy, which is essentially an independent-increment model, has the
hardest time fitting the observed option prices.  The more flexible but parametric CRC
model improves  on that fit, but still leaves 65.21\% of
quotes outside the spread; neither poor static fit is interpreted as a
dynamic no-arbitrage rejection.

\subsection{Risk-Neutral Diagnostics}
\label{sec:numerics-audit-definition}

Numerically, it is clear that as the number
of instruments increases on a fixed set of paths, the likelihood of a risk-free return
of a linear combination of instruments simply due to sampling error increases.

We will therefore test absence of dynamic arbitrage on a reduced trading universe which consists of spot and
44 options with cash strikes fixed at the initial date: eleven OTM calls or
puts at each of 5, 10, 15, and 20 business days.  Thus 45 instruments are live
initially, giving 20 spot and 550 option day--instrument coordinates over the
20-day horizon. 

The resulting test will only confirm absence of dynamic arbitrage within this restricted set of instruments, 
and only within the class of strategies considered. To illustrate the limitations of the former,
we apply the same test to an enlarged trading universe
in Section~\ref{sec:numerics-full-instrument-grid}.

\paragraph{Unconditional diagnostic.}
For every day--instrument coordinate $(n,i)$, first average the two members of
each antithetic pair and then form
\begin{equation}\label{eq:numerics-z-score}
  z_{n,i}
  =\frac{\overline R_{n,i}}
  {\widehat{\operatorname{se}}(\overline R_{n,i})}.
\end{equation}
The tables report the maximum and 99th percentile of $\abs{z_{n,i}}$ and the
number of coordinates beyond fixed thresholds.  This checks the unconditional
implication $\E^\Q[R_{n+1}^i]=0$ separately at every day and instrument.  It
does not test $\E^\Q[R_{n+1}^i\mid\mathcal F_n]=0$ at every simulated state.
Under~$\PP$ the same zero-mean statistics describe fitted physical drift; they
are not no-arbitrage violations.  For one standard-normal coordinate the 99th
percentile of $\abs{z}$ is 2.576.  Coordinate dependence and the simultaneous
screen mean that the reported counts are diagnostics, not family-wise
$p$-values.  The $1.96$ and $2.576$ thresholds apply to each coordinate-level
$z$-score, not directly to the maximum or empirical quantiles across all
coordinates.

\paragraph{Conditional diagnostics.}
We split complete antithetic pairs into two equal path sets $A$ and $B$, and
write $\mathcal I_A$ for all live path--time observations from paths in $A$.
The state $X_{p,n}$ contains the model's current predictable state---the
cumulative Brownian-factor histories for GROMM and Tangent L\'evy, or current
variance and structural parameters for CRC---together with log spot, the live
log option-price-to-spot ratios, calendar time, and time to each expiry.  On
set $A$ we fit the uncapped affine-plus-deep action
\begin{equation}\label{eq:numerics-action-map}
  a_\theta(X_n)
  =D_n\left[b+L\widetilde X_n+g_\theta(\widetilde X_n)\right],
\end{equation}
where $D_n$ masks expired instruments, $\widetilde X$ is standardized using
set $A$ only, and $g_\theta$ is a two-layer width-32 SiLU network.

For either diagnostic $\delta\in\{\mathrm{sqr},\mathrm{ent}\}$, let
$\widehat a_{-p}^{\,\delta}$ denote the action fitted on the fold not
containing path $p$.  Its out-of-fold score is
\begin{equation}\label{eq:numerics-action-path-score}
  d_p^\delta
  =\frac1H\sum_{n=0}^{H-1}
    \widehat a_{-p}^{\,\delta}(X_{p,n})^\top R_{p,n},
\end{equation}
where $H=20$ in the restricted experiment and $H=25$ on the full instrument
grid.  The factor $1/H$ only fixes the unit: removing it multiplies both the
estimate and its standard error by $H$.  For every antithetic pair
$r=(p,\bar p)$ define
$\widetilde d_r^\delta=(d_p^\delta+d_{\bar p}^\delta)/2$, let
$\overline d^\delta$ be the mean of the $N_{\mathrm{pair}}$ pair scores, and
let $s_\delta$ be their sample standard deviation.  We use the following two
diagnostics, in this order.

\begin{enumerate}[label=(\roman*)]
\item \emph{Squared diagnostic.}
  On training fold $A$ we maximize
  \begin{equation}\label{eq:numerics-squared-objective}
    \widehat J_A^{\mathrm{sqr}}(\theta)
    =\frac{1}{\abs{\mathcal I_A}}
      \sum_{(p,n)\in\mathcal I_A}
      \left[
        a_\theta(X_{p,n})^\top R_{p,n}
        -\frac12\lVert a_\theta(X_{p,n})\rVert^2
        -\frac{\lambda}{2}
          \bigl(a_\theta(X_{p,n})^\top R_{p,n}\bigr)^2
      \right],
      \qquad \lambda=1.
  \end{equation}
  Pointwise in the state, its unrestricted population optimizer is
  $(I+\lambda\E[RR^\top\mid X])^{-1}\E[R\mid X]$.  Thus the realized squared
  portfolio return includes cross-instrument second moments without requiring
  a separately estimated covariance matrix.  The squared action statistic is
  \begin{equation}\label{eq:numerics-t-sqr}
    T_{\mathrm{sqr}}
    :=\frac{\overline d^{\mathrm{sqr}}}
      {s_{\mathrm{sqr}}/\sqrt{N_{\mathrm{pair}}}}.
  \end{equation}

\item \emph{Entropy diagnostic.}
  On training fold $A$ we minimize
  \begin{equation}\label{eq:numerics-entropy-objective}
    \widehat L_A^{\mathrm{ent}}(\theta)
    =\log\!\left[
      \frac{1}{\abs{\mathcal I_A}}
      \sum_{(p,n)\in\mathcal I_A}
      \exp\!\left(
        -a_\theta(X_{p,n})^\top R_{p,n}
        +\frac12\lVert a_\theta(X_{p,n})\rVert^2
      \right)
    \right].
  \end{equation}
  For conditionally Gaussian returns the population action is
  approximately\footnote{On the filtration generated by the tangent space of
  the optimal action.}
  \[
    a_{\mathrm{ent}}^\star(X_t)
    \approx (I+\operatorname{Cov}_t(R))^{-1}\E_t[R].
  \]
  The entropy action statistic is
  \begin{equation}\label{eq:numerics-t-ent}
    T_{\mathrm{ent}}
    :=\frac{\overline d^{\mathrm{ent}}}
      {s_{\mathrm{ent}}/\sqrt{N_{\mathrm{pair}}}}.
  \end{equation}
\end{enumerate}

In each case the action is fitted on $A$, frozen and evaluated on untouched
$B$, after which $A$ and $B$ are exchanged and the two out-of-fold sets of
pair scores are combined.  Under the conditional zero-mean null hypothesis
every frozen predictable action has zero expected score, and both
$T_{\mathrm{sqr}}$ and $T_{\mathrm{ent}}$ are approximately standard normal.
Either statistic is shown in black when its absolute value is at most $1.96$,
orange when it lies in $(1.96,2.576]$, and red when it exceeds $2.576$.
These are per-test confidence thresholds.  Each is a conditional test relative
to the state and strategy class in~\eqref{eq:numerics-action-map}, rather than
a proof over all measurable strategies.

For broad instrument universes, isolated large $z$-scores are expected under
the null purely from multiple testing.  We therefore report both the extremes
and the number of coordinates beyond fixed thresholds, and compare the latter
with the corresponding standard-normal frequency.

\subsection{Restricted Universe}

% Generated by scripts/generate_main2_crc_q_audit.py; do not edit.
\providecommand{\CrcQPathCount}{100000}
\providecommand{\CrcQPairCount}{50000}
\providecommand{\CrcQHorizonDays}{20}
\providecommand{\CrcQStepsPerDay}{8}
\providecommand{\CrcResetGateSamples}{1000}
\providecommand{\CrcResetGateAccepted}{0}
\providecommand{\CrcResetGateTolerance}{2\times10^{-6}}
\providecommand{\CrcResetGateMedianError}{3.36\times10^{-4}}
\providecommand{\CrcResetGateMaximumError}{1.53\times10^{-3}}
\providecommand{\CrcPopulationMaximumResetError}{2.91\times10^{-3}}
\providecommand{\CrcQCoordinates}{570}
\providecommand{\CrcQMaximumAbsZ}{140.250}
\providecommand{\CrcQPercentileNineNineAbsZ}{110.774}
\providecommand{\CrcQAboveTwoFiveSevenSix}{371}
\providecommand{\CrcQAboveFive}{320}
\providecommand{\CrcQSquaredT}{42.562}
\providecommand{\CrcQEntropyT}{209.845}

Table~\ref{tab:numerics-restricted-q} shows the $\Q$ results for the restricted
trading universe.  Orange marks a moderate departure from the Gaussian
benchmark and red would mark a five-standard-error violation in the
coordinate-wise screen.  The separate conditional-statistic rule is given
above.  All eight GROMM conditional entries and both Tangent L\'evy conditional
entries are black.

\begin{table}[H]
\centering
\small
\begin{tabular}{lrrrrrrr}
\toprule
Model & $N$ & Max. $\abs{z}$ & 99\% $\abs{z}$
  & $N_{2.576}$ & $N_5$ & $T_{\mathrm{sqr}}$ & $T_{\mathrm{ent}}$\\
\cmidrule(lr){3-6}
\cmidrule(lr){7-8}
 & & \multicolumn{4}{c}{Unconditional} & \multicolumn{2}{c}{Conditional}\\
\midrule
A: rank 8, one step    & 570 & 2.716 & 2.464 &  3 & 0 & 1.394 & 1.401\\
B: rank 16, one step   & 570 & {\color{orange}3.663} & {\color{orange}3.618}
  & {\color{orange}10} & 0 & 1.037 & 1.040\\
C: rank 8, eight steps & 570 & 2.793 & 2.373 &  4 & 0 & 1.755 & 1.753\\
D: rank 16, eight steps& 570 & 2.996 & {\color{orange}2.967}
  & {\color{orange}22} & 0 & 0.039 & 0.040\\
Tangent L\'evy & 570 & 2.191 & 2.126
  & 0 & 0
  & $-0.768$ & $-0.575$\\
CRC  & \CrcQCoordinates & {\color{red}\CrcQMaximumAbsZ} & {\color{red}\CrcQPercentileNineNineAbsZ}
  & {\color{red}\CrcQAboveTwoFiveSevenSix} & {\color{red}\CrcQAboveFive}
  & {\color{red}\CrcQSquaredT} & {\color{red}\CrcQEntropyT}\\
\bottomrule
\end{tabular}
\caption{Restricted-universe, 100,000-path $\Q$ diagnostics.  The first five
numeric columns are coordinate-wise unconditional statistics; the final two
are the two-fold out-of-fold conditional statistics.  Here $N_c$ is the number
of coordinates with $\abs{z}>c$.  The red CRC row reflects that with our trained
data the model prices did not  match at the model's reset dates.}
\label{tab:numerics-restricted-q}
\end{table}

For the four GROMM specifications, no restricted coordinate exceeds five
standard errors and no conditional statistic exceeds 1.755 in absolute value.
Model~D has the richest factor representation, the finer time discretization,
and the conditional statistic closest to zero.  The repaired Tangent L\'evy
implementation is also consistent with the risk-neutral null: its maximum and
99th-percentile coordinate statistics are 2.191 and 2.126, no coordinate
exceeds 2.576, and the two conditional statistics are
$T_{\mathrm{sqr}}=-0.768$ and $T_{\mathrm{ent}}=-0.575$.  Both corresponding
95\% confidence intervals contain zero.  These results use the paper-aligned
daily-grid state and constrained reduced-node diffusion described in
Appendix~\ref{app:tangent-levy}; the poor initial market fit reported above is
therefore distinct from the absence of a rejection in these dynamic tests.
In contrast, the generalized-Bates CRC implementation fails the no-arbitrage tests
in our implementation and calibration of the model.
Its learned daily reset map fails the independent exact
repricing of the options on all \CrcResetGateSamples\ tested resets.\footnote{The test was performed at tolerance
  $\CrcResetGateTolerance$ in spot-normalized price (median maximum error
$\CrcResetGateMedianError$; maximum $\CrcResetGateMaximumError$).}  The resulting
\CrcQPathCount\ paths
therefore have artificial resets in option prices which our tests detect:
\CrcQAboveTwoFiveSevenSix\ of \CrcQCoordinates\ coordinates exceed 2.576 standard
errors, \CrcQAboveFive\ exceed five, and $T_{\mathrm{sqr}}=\CrcQSquaredT$,
$T_{\mathrm{ent}}=\CrcQEntropyT$.  This is somewhat expected as our model as calibrated does not satisfy
condition~(i) of Definition~6.3 in~\cite{DysLitGambaraTeichmann2021}.

\begin{remark}
We also applied our no-arbitrage tests to the statistical GROMM measure.
In this case we \emph{expect} the tests to fail.
Table~\ref{tab:numerics-restricted-p} confirms the failure:
all 570 coordinates exceed five standard errors under every physical-measure
specification, and the out-of-fold entropy statistics range from 89.205 to
646.645.  Thus the same diagnostics that remain small under~$\Q$ recover both
unconditional and state-predictable drift after the fitted change to~$\PP$.
\begin{table}[H]
\centering
\small
\begin{tabular}{lrrrrrr}
\toprule
Model & $N$ & Max. $\abs{z}$ & 99\% $\abs{z}$
  & $N_{2.576}$ & $N_5$ & $T_{\mathrm{ent}}$\\
\cmidrule(lr){3-6}
 & & \multicolumn{4}{c}{Unconditional} & \\
\midrule
A: rank 8, one step    & 570 & 62.138 & 59.780 & 570 & 570 & 90.038\\
B: rank 16, one step   & 570 & 342.470 & 328.811 & 570 & 570 & 646.645\\
C: rank 8, eight steps & 570 & 61.002 & 58.398 & 570 & 570 & 89.205\\
D: rank 16, eight steps& 570 & 336.819 & 325.038 & 570 & 570 & 624.560\\
\bottomrule
\end{tabular}
\caption{Restricted-universe, 100,000-path statistical measure $\PP$ diagnostics, reported against
zero to display the predictable drift created by the fitted
measure change.}
\label{tab:numerics-restricted-p}
\end{table}
\end{remark}

\subsection{Full Instrument Grid}
\label{sec:numerics-full-instrument-grid}

The preceding experiment deliberately trades few options.  We now test
Model~C on spot and the complete fixed grid of 200 interior strikes at each of
25 expiries, for 5,001 possible instruments.  Eligibility is predictable and
path dependent: before each daily trade, a live option is included only when
\[
  -2 < \frac{\log(K/Z_{p,n})}
  {\widehat\sigma_{\mathrm{ATM}}\sqrt{T_j-t_n}} < 1,
  \qquad \widehat\sigma_{\mathrm{ATM}}=9.52\%.
\]
Thus short expiries contain fewer eligible strikes and the traded set changes
with spot.  Ineligible and expired actions are exactly zero.  Each conditional
policy maps 5,035 pre-trade state variables to the masked 5,001-dimensional
action.  Each path score in~\eqref{eq:numerics-action-path-score} is averaged
over 25 rather than 20 daily transitions; this changes its units but not either
conditional statistic.

\begin{table}[H]
\centering
\small
\begin{tabular}{lrrrrrrrr}
\toprule
Model & $N_{\mathrm{inst}}$ & $N_{\mathrm{coord}}$
  & Max. $\abs{z}$ & 99\% $\abs{z}$
  & $N_{2.576}$ & $N_5$ & $T_{\mathrm{sqr}}$ & $T_{\mathrm{ent}}$\\
\cmidrule(lr){4-7}
\cmidrule(lr){8-9}
 & & & \multicolumn{4}{c}{Unconditional} & \multicolumn{2}{c}{Conditional}\\
\midrule
C: rank 8, eight steps
  & 5,001 & 62,569 & {\color{orange}4.140} & 2.583 & 635 & 0 & 0.843 & 0.843\\
\bottomrule
\end{tabular}
\caption{Full-instrument-grid, 100,000-path $\Q$ diagnostics.  A daily
coordinate enters the unconditional screen when it has at least 20 eligible
path observations and positive Monte Carlo variance.  Of 63,180 coordinates
that are eligible on at least one path, 62,569 enter the screen.}
\label{tab:numerics-full-instrument-grid}
\end{table}

At the two-sided 1\% threshold, 625.69 exceedances are expected among 62,569
independent standard-normal statistics; the observed count is 635.  The 188
coordinates above three standard errors are similarly close to the 168.9
expected count, and none exceeds five.  Dependence across neighbouring strikes
and expiries makes these benchmark counts descriptive rather than a formal
joint test, but the unconditional screen shows no excess tail beyond ordinary
multiple-testing variation.

The squared diagnostic gives a mean captured daily return of
$3.64\times10^{-8}$ spot units with pair-based standard error
$4.32\times10^{-8}$, hence $T_{\mathrm{sqr}}=0.843$.  Its 95\% confidence
interval, $[-4.83\times10^{-8},1.21\times10^{-7}]$, contains zero.  The entropy
diagnostic gives the same rounded mean, standard error, confidence interval,
and $T_{\mathrm{ent}}=0.843$.  The unconditional and both
strategy-conditional diagnostics therefore remain consistent with risk
neutrality on the full dynamically eligible instrument grid.

\section{Conclusion}

We have given an explicit continuous-time construction of a risk-neutral measure for the SANOS
discrete local volatility representation of European option markets, extending the static
SANOS surface of~\cite{SANOS2026} to a dynamically arbitrage-free model. The resulting SDLV
process admits a closed-form risk-neutral drift for every discrete local volatility coordinate
and an exact front-density/spot diffusion, in place of the numerical martingale measure
used in~\cite{DYSANOS2026}. We further proved that the DLV dynamics admit a unique strong solution under
the assumption of a square-integrable bound on the vol-of-vol (Theorem~\ref{thm:existence_bh}).

Estimated on 891 trading days of S\&P~500 index option data, the model's static fit covers more
than 95\% of quoted bid-ask intervals, and the restricted-universe $\Q$ diagnostics of
Section~\ref{sec:numerics} confirm the theoretical martingale property up to time-discretization
error.  On the full dynamically eligible instrument grid, both the
coordinate-wise unconditional screen and the two conditional
diagnostics remain consistent with the risk-neutral null.
We further estimate the physical dynamics and the implied market price of risk, recovering a
stable low-rank factor structure for the DLV vol-of-vol.

We reviewed the two published related models, Tangent L\'evy~\citep{TangentLevy} and CRC~\citep{DysLitGambaraTeichmann2021}. The former is free of dynamic arbitrage
but does not fit well to observed option prices. Our implementation of the latter does not fit observed option prices
well, nor did we manage to construct a risk-neutral regridding process as stipulated in~\cite{DysLitGambaraTeichmann2021}.

Natural next steps include applying the full-instrument diagnostic to the
remaining factor ranks and time discretizations, studying the discrete-branch
correction of Remark~\ref{rem:discrete-branch-sampling}, and extending
the estimation beyond short-dated S\&P~500 index options to other asset classes and longer
tenors.

\section*{Acknowledgements}

We want to thank the team at Wharton for their WRDS \url{https://wrds-www.wharton.upenn.edu} portal and OptionMetrics \url{https://optionmetrics.com/united-states} 
for providing their IvyDB data via this service. Without this data our work would not have been possible.

\section*{Declaration of funding}
BH gratefully acknowledges financial support from UKRI1010: High order mathematical and computational infrastructure for streamed data that enhance contemporary generative and large language models.

\appendix
% !TEX root = main2.tex
\section{Model and DLV factor visualizations}

\subsection{Consistent Drift}
\label{app:drift}

Figure~\ref{fig:numerics-best-paths} illustrates the fitted measure change
using four matched Model~D paths. 
Solid and dashed lines of the same colour use
the same Gaussian draws under~$\Q$ and~$\PP$.  Besides spot, the illustration
reports ATM one- and five-day implied volatility and the central log-strike skew
\begin{equation}\label{eq:numerics-atm-skew}
  \operatorname{Skew}_{\tau}
  =\frac{
    \sigma_{\mathrm{imp}}(K=Z_t e^h,\tau)
    -\sigma_{\mathrm{imp}}(K=Z_t e^{-h},\tau)
  }{2h},
  \qquad h=0.01.
\end{equation}
The paths illustrate the fitted measure change; four paths are not a physical
model validation.

\begin{figure}[H]
  \centering
  \includegraphics[height=0.88\textheight]
    {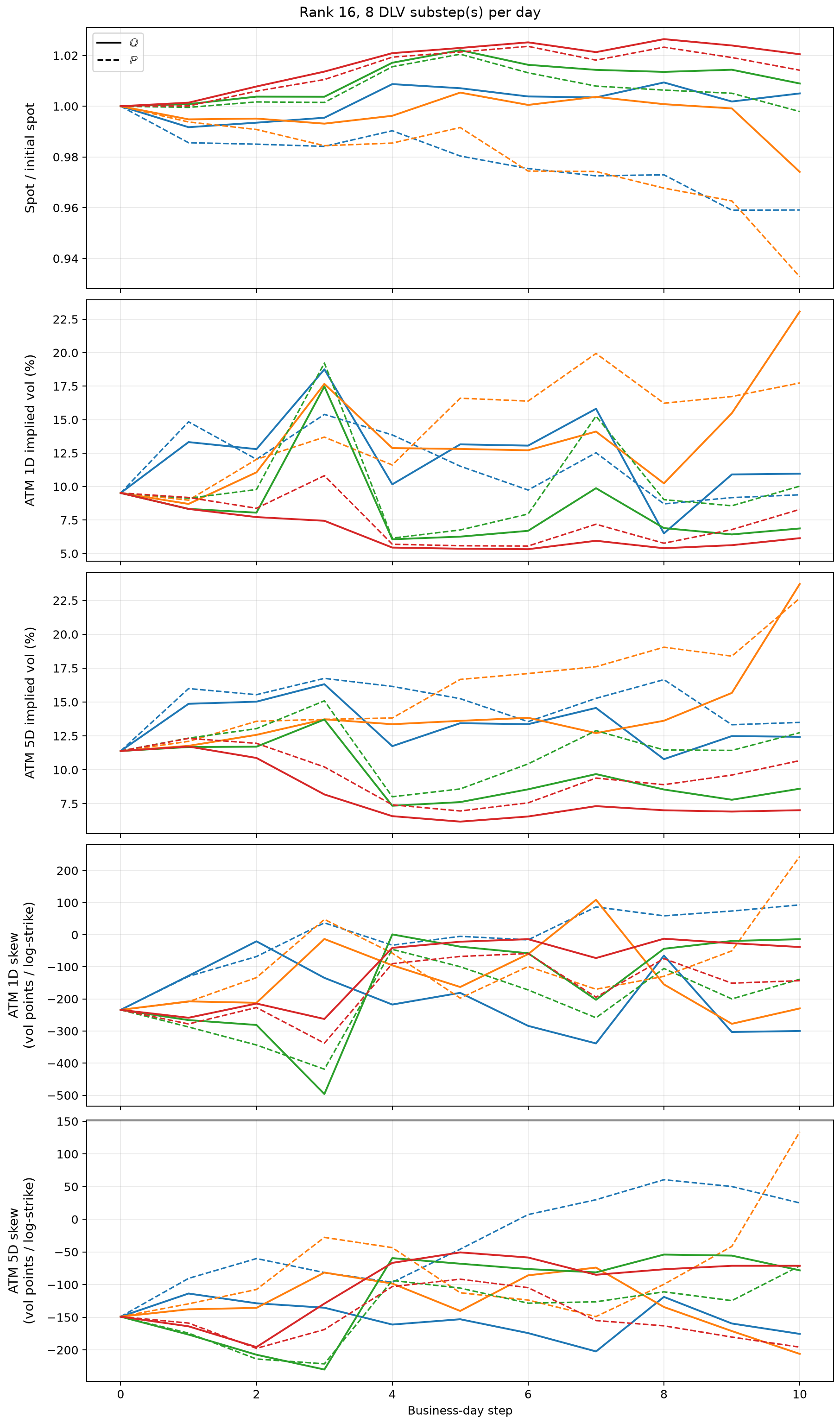}
  \caption{Matched $\PP/\Q$ paths for Model~D: rank 16 with eight Brownian
  substeps per day.}
  \label{fig:numerics-best-paths}
\end{figure}

\subsection{PCA Summary}
\label{app:dlv-factor-figures}

The remaining figures visualize the factor systems estimated in
Section~\ref{sec:estimation-vol-of-vol}.  Figure~\ref{fig:numerics-baseline-iv-3d}
is the baseline risk-neutral implied-volatility surface used for all factor
illustrations.

\begin{figure}[H]
  \centering
  \includegraphics[width=0.72\textwidth]
    {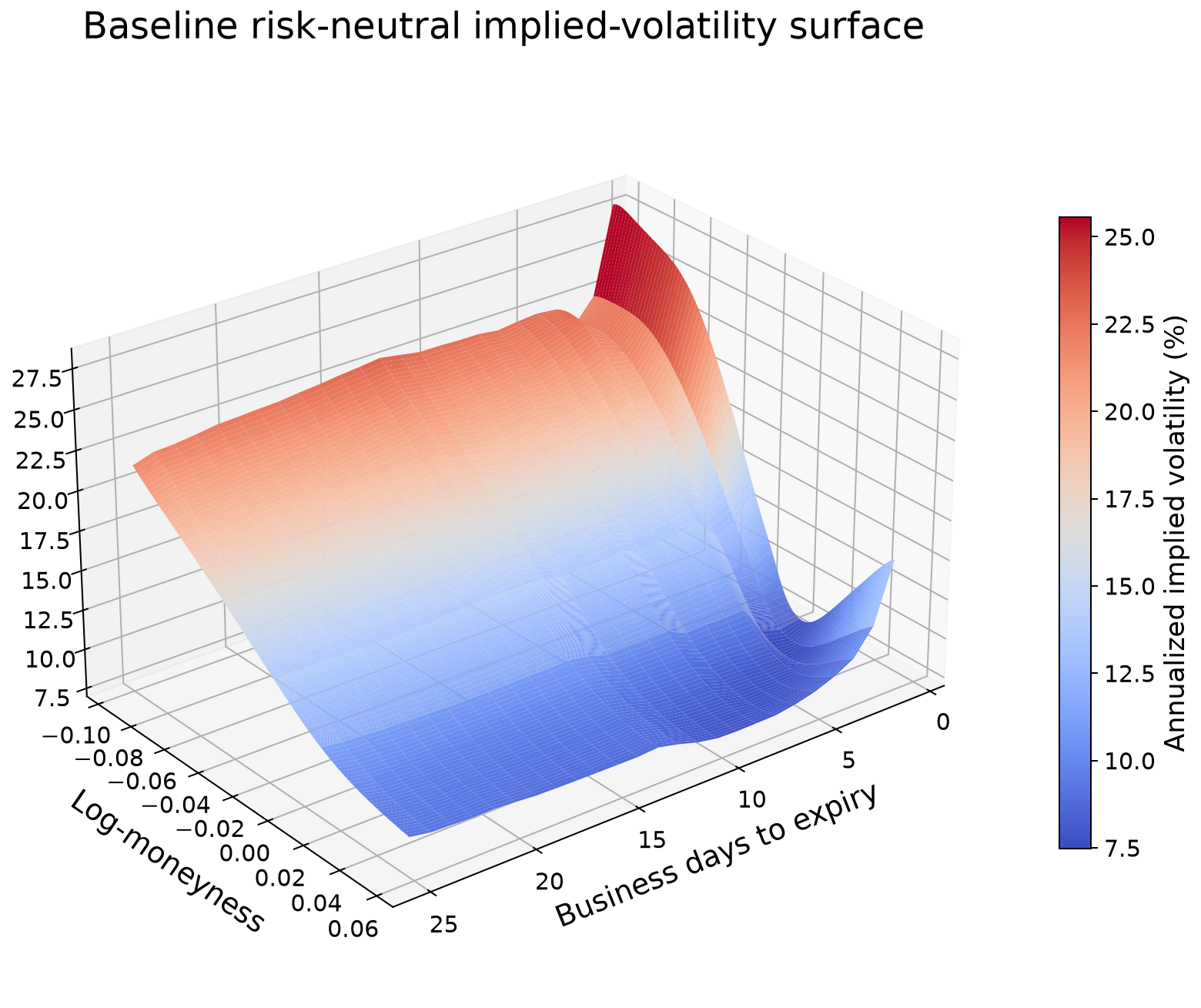}
  \caption{Baseline annualized implied volatility on the fitted
  log-moneyness and business-days-to-expiry grid.}
  \label{fig:numerics-baseline-iv-3d}
\end{figure}

Figure~\ref{fig:numerics-factor-power} reports individual and cumulative PCA
power.  The experiments retain either eight or 16 factors; the first rank
above 90\% is 26.

\begin{figure}[H]
  \centering
  \includegraphics[width=0.88\textwidth]
    {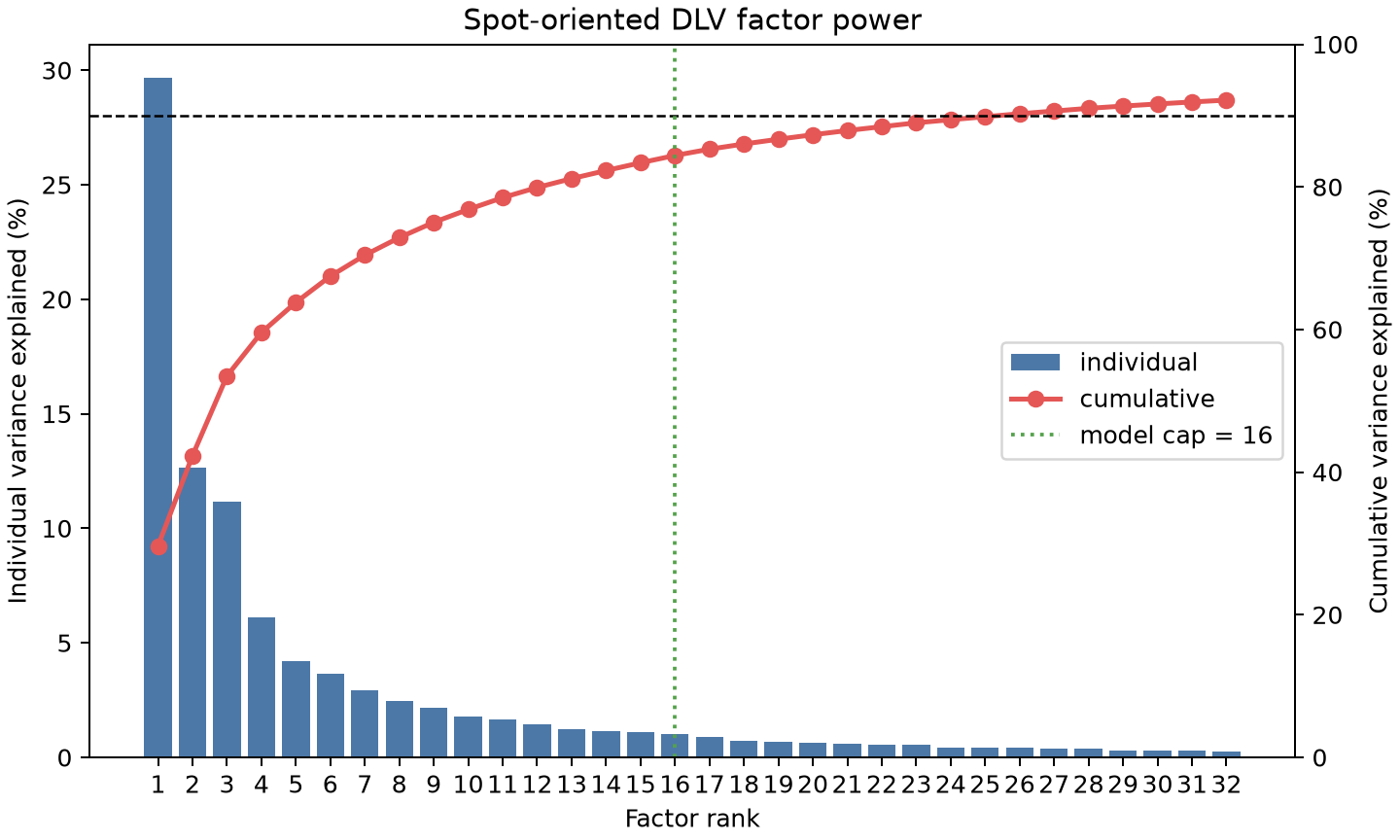}
  \caption{Individual and cumulative variance explained by the spot-oriented
  DLV factor system.  The vertical line is the rank-16 model cap; 26 factors
  are required to exceed 90\%.}
  \label{fig:numerics-factor-power}
\end{figure}

The rank-16 fit appends PCs 9--16 to the rank-8 fit.  A one-standard-deviation
shock is the diffusion-only log-DLV displacement
$\pm\xi_f\sqrt{1/252}$.  The figures display its central implied-volatility
response
\begin{equation}
  \label{eq:numerics-iv-factor-shock}
  \frac12\left[
    \sigma_{\mathrm{imp}}(\log\Sigma+\xi_f/\sqrt{252})
    -\sigma_{\mathrm{imp}}(\log\Sigma-\xi_f/\sqrt{252})
  \right].
\end{equation}
These shock directions are the same under $\PP$ and $\Q$: an equivalent
measure change alters drift, not diffusion.  Each panel deliberately retains
its own vertical scale.

\begin{figure}[H]
  \centering
  \includegraphics[width=\textwidth]
    {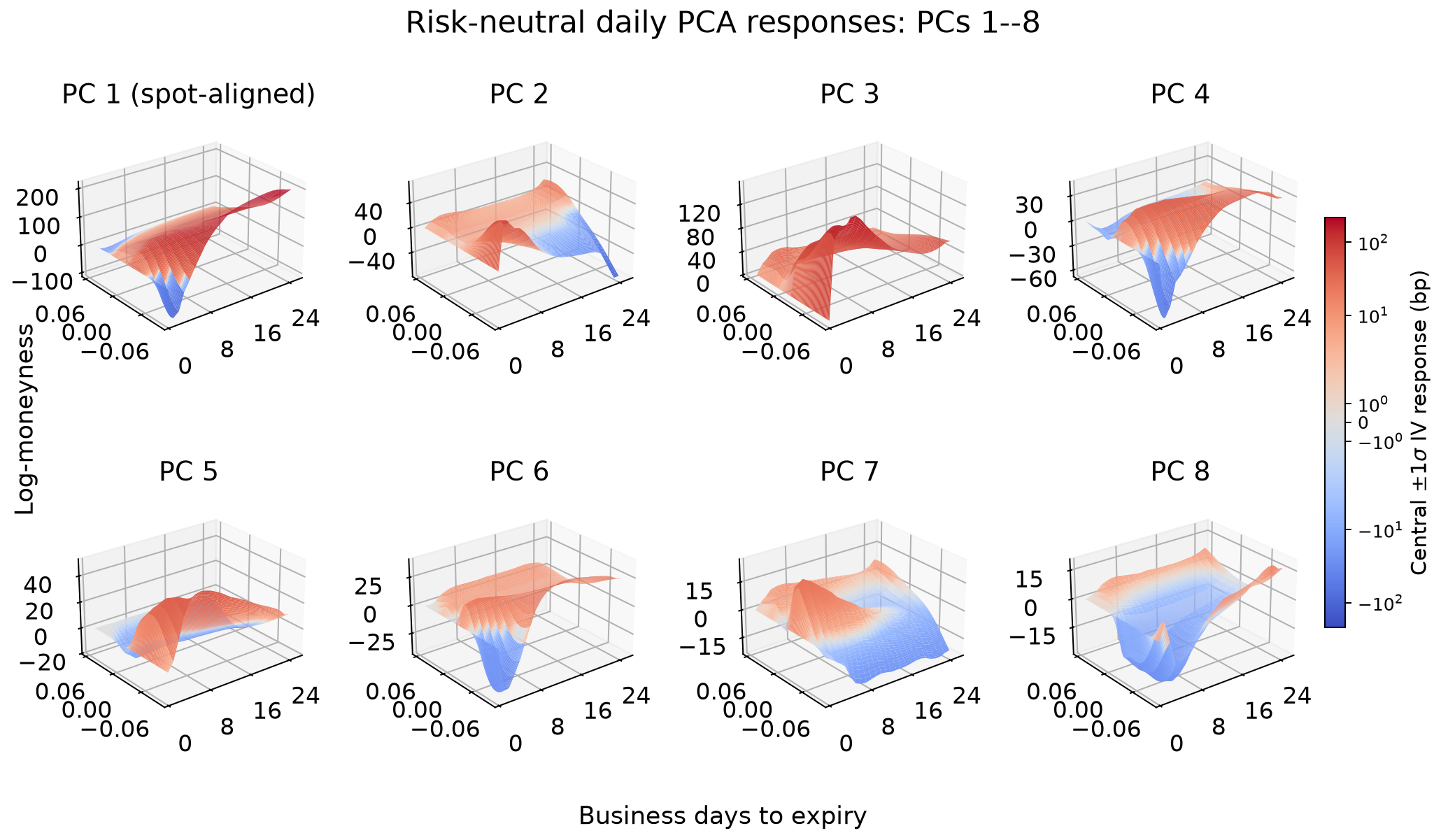}
  \caption{Risk-neutral one-standard-deviation implied-volatility responses
  for PCs 1--8.  PC~1 is aligned to the front-spot score.}
  \label{fig:numerics-pc1-8}
\end{figure}

\begin{figure}[H]
  \centering
  \includegraphics[width=\textwidth]
    {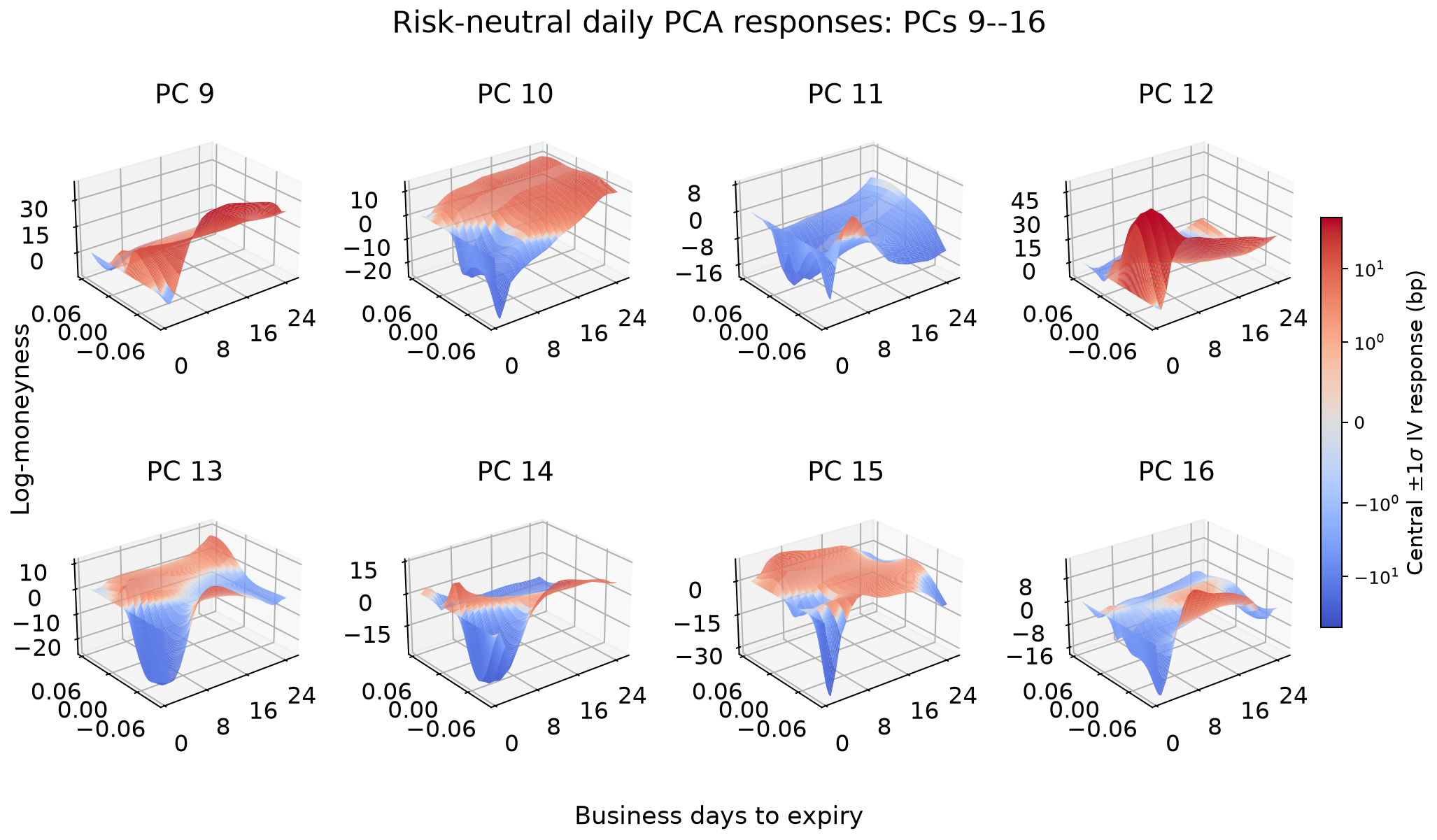}
  \caption{The additional one-standard-deviation implied-volatility responses,
  PCs 9--16, in the nested rank-16 model.}
  \label{fig:numerics-pc9-16}
\end{figure}

% !TEX root = main2.tex
\section{Market-data and calibration pipeline}
\label{app:market-pipeline}

This appendix records the data and numerical pipeline behind
Section~\ref{sec:estimation}. It follows the broader SPX pipeline described in
\cite{DYSANOS2026}, restricted here to the short-dated fixed-expiry SANOS
experiment. The complete transformation is
\[
\begin{aligned}
\text{raw IvyDB quotes}
&\longrightarrow \text{daily cleaning and parity-forward enrichment}\\
&\longrightarrow \text{arbitrage-free SANOS eligibility surface}\\
&\longrightarrow \text{European short-horizon quote selection}\\
&\longrightarrow \text{daily 25-by-200 calibration tensor}\\
&\longrightarrow \text{direct or affine end-to-end fit}.
\end{aligned}
\]

\subsection{Raw data and cleaning}
\label{app:pipeline-raw}

The source is OptionMetrics IvyDB US data accessed through WRDS. For each
eligible SPX market date we load the index level and previous index level,
option strikes and calendar expiries, call and put bid and ask quotes, volume,
open interest, AM/PM settlement flags, and the zero curve. The retained sample
runs from 3 January 2022 through 29 August 2025 and contains 891 valid full-day
NYSE sessions.

The raw cleaning step removes known bad ticker--date observations, missing or
invalid quotes, invalid settlement flags, same-day expiries, irregular
expiries, and quotes whose bid does not exceed $0.011$. Business days and year
fractions to expiry are then computed; an AM-settled expiry is treated as
$0.2$ of a business day earlier than its corresponding PM expiry.

For each expiry we recompute the forward from put--call parity rather than use
the delayed IvyDB forward table. With discount factor $D$, define the
discounted-forward bounds at a strike $K$ by
\[
  G^{\mathrm{ask}}(K)=C^{\mathrm{ask}}(K)-P^{\mathrm{bid}}(K)+DK,
  \qquad
  G^{\mathrm{bid}}(K)=C^{\mathrm{bid}}(K)-P^{\mathrm{ask}}(K)+DK.
\]
Their midpoint is averaged across strikes with inverse-spread weights,
\[
  \widehat G
  =\frac{\sum_K w_K\,[G^{\mathrm{ask}}(K)+G^{\mathrm{bid}}(K)]/2}
         {\sum_K w_K},
  \qquad
  w_K=\frac1{\max\{G^{\mathrm{ask}}(K)-G^{\mathrm{bid}}(K),10^{-3}\}}.
\]
The resulting forward is used to construct the tightest call-equivalent quote,
forward-normalized strikes and prices, implied volatilities, and Greeks. An
expiry must contain at least ten strikes with a traded call or put and at least
four strikes with both sides traded. The normalized spread must exceed
$10^{-6}$ and retained implied volatilities must lie between $0.01$ and $5$.

\subsection{Static preprocessing and quote selection}
\label{app:pipeline-static}

Before applying the model-specific horizon, the pipeline retains quotes with
vega per unit total volatility of at least $10^{-3}$, log volume and log open
interest of at least one, and at least ten surviving options per expiry. The
nominal per-day quote cap is 100,000 and is therefore non-binding here.

For data-quality control and to construct a stable ATM curve, every daily book
is passed through the linear SANOS preprocessing of \cite{SANOS2026}. The
strike grid is extended, each expiry is fitted inside the quoted bid-ask
intervals, negligible-density redundant strikes are removed, and the joint
surface fit enforces both butterfly and calendar-spread constraints. The
smoothness parameter is $0.25$. Expiries and strikes for which this
arbitrage-free auxiliary fit is unavailable are not passed to the final
calibration.

This auxiliary SANOS surface determines eligibility, the ATM term structure,
and the floating model support. It is not the calibration target: the JAX fits
in Section~\ref{sec:estimation-fit} use the original retained call and put
bid-ask intervals. No linear-programming state or fitted midpoint is supplied
as a target.

The final selection keeps European quotes no farther than 25 business days
from the observation date, excludes the $0.2$-business-day slice, and retains
the standardized log-moneyness interval $[-2.5,1.5]$. This leaves 2,192,458
quotes, between 1,590 and 5,096 per date.

\subsection{Daily calibration tensor}
\label{app:pipeline-tensor}

The model expiry grid is
\[
  T_j=\frac{j}{252},\qquad j=1,\ldots,25.
\]
For each date, 100 interior-strike candidates are generated from the front
ATM scale and 100 from the 25-day ATM scale. They are merged with a minimum
log-spacing safeguard to form one 200-node support shared across all model
expiries on that date. The support floats from date to date. Two boundary
nodes, equal to one half of the lowest inner strike and one and a half times
the highest inner strike, complete the 202-node density grid.

Each retained market quote stores its call/put flag, normalized strike, bid,
ask, midpoint, business days to expiry, adjacent model-expiry indices, and
linear interpolation weight. The contemporaneous SPX level is retained as a
reference spot. The immutable calibration tensor pads each date to a multiple
of 128 quotes and carries a Boolean mask, so padding never contributes to the
objective.

At valuation time the DLV recursion produces the 25 discrete densities. The
fixed 5\% Black factor convolves each discrete density into a smooth call-price
curve, and the two adjacent expiry curves are linearly interpolated to the
market quote's exact expiry. Puts follow from put--call parity.

\subsection{Optimization and retained artifacts}
\label{app:pipeline-optimization}

The direct and affine parametrizations are those in
Section~\ref{sec:estimation-fit}. Direct states are initialized from the
date-specific ATM term structure and are constant across strike at
initialization. Affine means use the average of these initial states; codes and
loadings receive small centred Gaussian initial values. DLV squared is bounded
between $10^{-4}$ and $9$, with a differentiable barrier near the upper bound.

The loss in \eqref{eq:estimation-quote-objective} is differentiated over every
retained quote. Quote blocks of size 128 are used only to limit accelerator
memory: their gradients are summed before each Adam update, so this is an exact
full-sample gradient rather than stochastic minibatching. Each fit uses 300
initial epochs with a cosine learning-rate schedule from $2\times10^{-2}$ to
$10^{-3}$, followed by a 700-epoch continuation from $5\times10^{-3}$ to
$5\times10^{-5}$. Checkpoints store either all daily direct states or the
affine mean, loadings, and daily codes. The retained reports recompute every
quote from the final checkpoint and produce the statistics in
Table~\ref{tab:numerics-static-market-fit}.

\clearpage
\section{Reference-model implementation details}
\label{app:reference-models}

\subsection{Tangent L\'evy}
\label{app:tangent-levy}

Our Tangent L\'evy benchmark follows the code book construction of
\citet{TangentLevy}, but is an in-repository implementation rather than the
authors' software.  It is paper-aligned rather than a literal reproduction.
At each observed expiry we first fit a double-exponential jump density to
forward-normalized call-equivalent bid-ask intervals.  For each jump node, the
forward L\'evy density is then represented by 12 non-negative cubic B-spline
coefficients in maturity.  Analytic spline integrals map these coefficients to
the fitted average L\'evy measures; non-negative least squares with a curvature
penalty supplies a smooth positive reconstruction.  We evaluate it on a daily
maturity grid from zero to 70 business days and a 301-node jump grid on
$[-0.75,0.75]$.

The historical diffusion fit uses every other daily maturity and the seven
nonzero jump nodes with highest average training intensity.  Linear
prolongation maps this reduced grid to the full code book, and the reduced
factors are projected onto the discrete exponential-moment constraint before
PCA.  The leading 16 directions are estimated from 698 consecutive-session
changes ending on 31 December 2024, explain 95.23\% of reduced-grid variation,
and use 5\% covariance shrinkage.  Simulation integrates the full-grid
loadings into eight substep cells per business day, applies the exact cell
shift, computes the quadratic convolutional HJM drift on its complete support,
and drives spot with the compensated shortest-maturity jump law.  The maximum
fitted moment-constraint residual is $5.13\times10^{-16}$.

To retain non-negative finite-step densities, each path uses a predictable
scalar time change $\gamma$: factor loadings are multiplied by $\gamma$ and
the quadratic HJM drift by $\gamma^2$.  The scale is fixed before observing the
Gaussian increment by solving the exact quadratic positivity bound with
12 conditional-standard-deviation headroom.  This preserves the HJM identity
but can materially slow the fitted dynamics near the boundary.  In the
100,000-path audit the mean and minimum scales are 0.133 and 0.0264; no substep
has $\gamma<0.01$ and no path freezes.

The initial-fit results in Figures~\ref{fig:initial-fit-iv}
and~\ref{fig:initial-fit-spread} use the final daily-grid reconstruction
supplied to the simulator.  This distinction is important: independent
maturity fits are not a dynamically admissible initial code book.  The
reconstruction on 29 August 2025 did not pass its pre-specified repricing gate,
so the poor fit shown there is a reported numerical limitation, not suppressed
calibration error.  The
subsequent risk-neutral diagnostics therefore test the implemented projected
Q dynamics conditional on this retained initial state; they do not turn the
failed static reconstruction into a successful market calibration.

The restricted-universe audit uses 100,000 paths arranged as 50,000
antithetic pairs over 20 business days.  Across its 570 live
day--instrument coordinates, the maximum and 99th-percentile absolute
$z$-scores are 2.191 and 2.126, with no coordinate above 2.576.  The two-fold
out-of-fold conditional statistics are $T_{\mathrm{sqr}}=-0.768$ and
$T_{\mathrm{ent}}=-0.575$; both 95\% confidence intervals contain zero.

\subsection{Bates CRC}
\label{app:bates-crc}

Our CRC benchmark is an in-repository implementation of the generalized
Bates construction and Algorithm~1 of
\citet{DysLitGambaraTeichmann2021}, not the authors' software.  The state has
Heston stochastic variance and compensated compound-Poisson Gaussian log
jumps.  Ten increasing target maturities are grouped into five adjacent
pairs; each pair has its own jump mean and standard deviation, while the
Poisson intensity and Heston mean-reversion speed are held fixed across
recalibrations.  Carr--Madan inversion with damping parameter $1.5$ supplies
the option prices.  We retain a scalar direct solve as an inverse-map oracle
and subject the learned map to independent exact-pricer reset gates.  The path
simulation also records its fixed-option price discontinuities, so a failed
numerical inversion is not silently reported as an exact CRC reset.

All historical inputs end strictly before 29 August 2025.  From the 890
earlier sessions we estimate $\kappa=15.0$ and jump intensity
$\lambda=4.25197$ per year.  The target-date calibration then fits $V_0$,
$\theta$, $\xi$, $\rho$, and the five jump mean/standard-deviation pairs to
the same 2,662 quotes used in Figure~\ref{fig:initial-fit-iv}.  It converges to
$V_0=0.007474$, $\theta=0.0100$, $\xi=0.5000$, and $\rho=-0.9900$; the latter
three are on their calibration bounds, consistent with the limited static
fit shown in Figures~\ref{fig:initial-fit-iv}
and~\ref{fig:initial-fit-spread}.

The exogenous reset law is estimated from 764 daily proxy states from 7 July
2022 through 28 August 2025.  Each row uses a trailing 63-session window ending
on that date: $\theta$ is the median 20-business-day ATM implied variance,
$\xi$ is the RMS CIR-normalized variance innovation, and $\rho$ is the
spot/variance-innovation correlation.  We fit a Gaussian law to
$(\Delta\log\theta,\Delta\log\xi,\Delta\operatorname{atanh}\rho)$, cap each
relative update at 5\%, keep $\kappa$ fixed, and impose the Feller condition
after every update.  These are backward-only dynamic proxies, not daily
full-surface Bates calibrations.

We train the two inverse-recalibration maps on 600,000 synthetic
generalized-Bates surfaces generated on the target calibration's exact ten
maturities and 13 relative-strike nodes, using 240-point Gauss--Legendre
Fourier integration on $[0,600]$; 540,000 surfaces form the training split and
60,000 whole-shard observations are held out.  Of the 78 million generated
surface nodes, 193,840 (0.249\%) required clipping back inside the Black price
bounds; the generator rejects a surface if that correction exceeds
$5\times10^{-4}$ in normalized price.  NN1 maps the fourteen
model coordinates $(V,\theta,\xi,\rho,\mu_1,\ldots,\mu_5,
\sigma_1,\ldots,\sigma_5)$ to the 130 raw Black volatilities.  NN2 receives
the updated $(\theta,\xi,\rho)$ and the protected 130-node surface and returns
the ten jump coordinates through a sigmoid stretched to their training
bounds: $[-0.30,0.15]$ for each mean and $[0.005,0.30]$ for each standard
deviation.  Both maps are plain four-hidden-layer, width-1,024 ELU networks.  We
first train NN1 with raw implied-volatility MSE and then freeze it while
training NN2 through NN1, using batches of 1,000.  This preserves the
input--output construction of Algorithm~1, but does not reproduce the
authors' batch-normalized residual cell and is another reason we describe the
benchmark as paper-aligned rather than a literal software reproduction.

The restricted audit uses the same geometry as the other reference models:
100,000 paths arranged as 50,000 adjacent antithetic pairs, 20 business days,
eight substeps per day, and the fixed spot-plus-44-option universe.  Between
daily resets we use full-truncation Heston variance and compensated Gaussian
jumps from the shortest live jump bucket.  At each day boundary the old
code book is shifted by one day, its ten-by-13 continuation surface is repriced,
the backward-estimated parameter innovation is applied, and NN2 reconstructs
the new five jump means and standard deviations.  The implementation records
the resulting fixed-option price discontinuity explicitly before any
no-arbitrage statistic is accepted.

That gate fails for the retained NN2 checkpoint: none of
\CrcResetGateSamples\ independently drawn reset states meets the
$\CrcResetGateTolerance$ spot-normalized exact-pricing tolerance.  The median
maximum error across its calibration and inter-node held-out options is
$\CrcResetGateMedianError$ and the maximum is $\CrcResetGateMaximumError$.
We nevertheless retain the \CrcQPathCount-path restricted audit as a negative
control for the diagnostics.  It has maximum and 99th-percentile coordinate
$z$-scores \CrcQMaximumAbsZ\ and \CrcQPercentileNineNineAbsZ; \CrcQAboveTwoFiveSevenSix\ of
\CrcQCoordinates\ coordinates exceed 2.576 and \CrcQAboveFive\ exceed five.
The two-fold out-of-fold conditional statistics are
$T_{\mathrm{sqr}}=\CrcQSquaredT$ and $T_{\mathrm{ent}}=\CrcQEntropyT$.
These rejections detect the reset-repricing defect; they do not establish that
an exact consistent-recalibration model is arbitrageable.

\subsection{Matched initial-fit convention}

For the comparison we use the exact 2,662 quotes retained by the GROMM
25-expiry calibration on 29 August 2025.  GROMM is the unrestricted
$25\times200$ state from the
\texttt{direct-continuation-\allowbreak1000-v1} run after
1,000 end-to-end epochs with a 5\% annual Black backbone; this is the same
checkpoint as the direct fit in Table~\ref{tab:numerics-static-market-fit}.
Tangent L\'evy is repriced at precisely those strikes and expiries from its
daily-grid density, and CRC from its retained generalized-Bates target fit.
Put observations are converted to normalized
call-equivalent prices by parity before Black-volatility inversion.  If $b$,
$a$, and $h=(a-b)/2$ are the bid, ask, and half-spread, respectively, the
signed price surface reports
\[
  \widetilde e=
  \begin{cases}
    (C-a)/h,&C>a,\\
    0,&b\leq C\leq a,\\
    (C-b)/h,&C<b,
  \end{cases}
  \qquad e=|\widetilde e|.
\]
The volatility surface applies the same construction after transforming the
bid and ask to Black volatility.  Thus ``quotes with excess'' is the fraction
with $e>0$, and the conditional mean and median are computed only over those
failed quotes.  All three models use the same observations, spreads, axes, and
colour scales across all three panels.

\clearpage
\section{Option-surface Literature Classification}
\label{app:literature-classification}
\def\GROMMpaper{}% switch the shared literature tables to this paper's rows
% Literature tables for option-surface construction and dynamics.
% This file is included from dy21.tex after the introductory literature review.
Tables~\ref{tab:surface-interpolation-literature},
\ref{tab:nonparametric-surface-interpolation-literature},
\ref{tab:surface-representation-literature},
\ref{tab:inverse-calibration-literature}, and
\ref{tab:surface-dynamics-literature} separate four questions which are often
conflated: constructing an option or implied volatility surface, compressing or generating a
market state, solving an underlying-model inverse problem, and learning a
surface's evolution through time.  We use common classification
symbols throughout the tables:
\begin{itemize}[leftmargin=*,nosep]
  \item \textbf{S$\infty$}: a theorem- or architecture-level guarantee on the
        full continuous surface.
  \item \textbf{S$\Omega$}: the same guarantee only on a stated bounded fitting
        domain.
  \item \textbf{SG}: a strict guarantee only on the modelled finite grid.
  \item \textbf{A}: approximate, using penalties, reweighting, rejection, or
        empirical violation rates.
  \item \textbf{N}: no static-arbitrage guarantee.
  \item \textbf{U}: arbitrage freedom is claimed, but the available construction
        does not suffice to verify its continuum scope.
\end{itemize}
Please contact the authors for any suggestions, corrections, or additions.

% Entries introduced by the backward-reference audit use the same styling.
\newcommand{\auditnew}[1]{#1}

% Paper switch: the GROMM/SDLV paper (main2.tex) defines \GROMMpaper before
% \input-ing this file; the DYSANOS papers do not and get their original rows.
\ifdefined\GROMMpaper
\newcommand{\OwnRowsRepresentation}{%
\citet{DYSANOS2026}
  & Smooth SANOS plus affine latent state & $18\times200=3{,}600$ prices; 3,620 controls & $\infty$ & S$\infty$ & The decoded continuum remains in the SANOS admissible family. \\
This paper
  & Smooth SANOS and stochastic DLV/PCA factors & $25\times200=5{,}000$ options plus spot & $\infty$ & S$\infty$ & Full smooth continuum representation. \\}
\newcommand{\OwnRowsDynamics}{%
\citet{DYSANOS2026}
  & SANOS, affine latent state, and AR(1) & S$\infty$ & Yes & Limited SG: 48 options plus spot & Yes & The implemented $\Q$ construction is finite-dimensional. \\
This paper
  & Smooth SANOS and stochastic DLV/PCA factors & S$\infty$ & Yes & S$\infty$ in strike at the 25 model expiries; explicit continuous-time DLV drift & Yes & Direct smooth full-surface construction. \\}
\else
\newcommand{\OwnRowsRepresentation}{%
This paper
  & Smooth SANOS plus affine latent state & $18\times200=3{,}600$ prices; 3,620 controls & $\infty$ & S$\infty$ & The decoded continuum remains in the SANOS admissible family. \\}
\newcommand{\OwnRowsDynamics}{%
This paper
  & SANOS, affine latent state, and AR(1) & S$\infty$ & Yes & Limited SG: 48 options plus spot & Yes & The implemented $\Q$ construction is finite-dimensional. \\}
\fi
\newcolumntype{L}[1]{>{\raggedright\arraybackslash}p{#1}}

\begin{landscape}
\begin{table}[p]
\centering
\fontsize{7}{7.8}\selectfont
\setlength{\tabcolsep}{3pt}
\renewcommand{\arraystretch}{1.04}
\caption{Selected parametric implied-volatility and option-price interpolation
methods.  ``Parametric'' means that the surface is restricted to a prescribed
finite-dimensional functional family; it does not describe the calibration
algorithm. Constraints imposed during fitting are distinguished from post-fit repair.}
\label{tab:surface-interpolation-literature}
\begin{tabular}{@{}L{4.2cm}L{3.8cm}L{2.2cm}L{3.2cm}L{2.0cm}L{8.0cm}@{}}
\toprule
Method & Representation & No static arb. & Native inputs & Off-grid & Scope and qualification \\
\midrule
\auditnew{\citet{Gatheral2004SVI}}
  & \auditnew{SVI smile} & \auditnew{N by itself}
  & \auditnew{Arbitrary, one expiry} & \auditnew{$\infty$ per expiry}
  & \auditnew{Introduces raw SVI and its wing asymptotics, but does not supply the later cross-expiry SSVI construction or a complete admissible parameter domain.} \\
\citet{Hafner2004}
  & Polynomial IV surface & N & Variable daily DAX transactions & $\infty$ on bounded domain
  & Weighted least-squares fit followed by finite-grid screening: dates violating hedge, bull-spread, butterfly, or calendar tests on a 200-point grid are discarded.  The fitted family itself is unconstrained. \\
\citet{Bloch2013MixVar}
  & Mixture family & U & Arbitrary & $\infty$
  & Claimed arbitrage-free by construction; the continuum conditions were not verified. \\
\citet{GatheralJacquier2014}
  & SVI/SSVI & S$\infty$ & Arbitrary & $\infty$
  & Strict for admissible parameters; a central low-dimensional benchmark. \\
\citet{GuoJacquierMartiniNeufcourt2016}
  & Generalized SVI & S$\infty$ & Arbitrary & $\infty$
  & Extends the strict SVI/SSVI constructions. \\
\auditnew{\citet{CorbettaEtAl2019}}
  & \auditnew{SSVI slices} & \auditnew{S$\infty$}
  & \auditnew{Arbitrary maturities and strikes} & \auditnew{$\infty$}
  & \auditnew{Calibrates butterfly- and calendar-consistent slices, then proves that its parameter interpolation and extrapolation preserve those conditions over the continuous surface.} \\
\auditnew{\citet{HendriksMartini2019}}
  & \auditnew{eSSVI} & \auditnew{S$\infty$ (conditional)}
  & \auditnew{Arbitrary} & \auditnew{$\infty$}
  & \auditnew{Explicit butterfly and calendar conditions; strictness requires admissible parameter curves.} \\
\auditnew{\citet{MartiniMingone2022}}
  & \auditnew{SVI smile} & \auditnew{S$\infty$ per expiry}
  & \auditnew{Arbitrary, one expiry} & \auditnew{$\infty$ per expiry}
  & \auditnew{Characterizes the no-butterfly domain of raw SVI and calibrates directly on it; calendar consistency across independently fitted expiries is not automatic.} \\
\auditnew{\citet{Mingone2022}}
  & \auditnew{eSSVI surface} & \auditnew{S$\infty$}
  & \auditnew{Arbitrary} & \auditnew{$\infty$}
  & \auditnew{A global eSSVI parametrization whose admissible parameter domain enforces butterfly and calendar-spread conditions.} \\
\citet{Roos2025}
  & Discrete-time local-vol family & U & Arbitrary coterminal set & $\infty$
  & The abstract claims analytic strike interpolation and expiry extension; the continuum inequalities were not verified. \\
\bottomrule
\end{tabular}
\end{table}
\end{landscape}

\begin{landscape}
\fontsize{7}{7.7}\selectfont
\setlength{\tabcolsep}{2.5pt}
\renewcommand{\arraystretch}{1.04}
\begin{longtable}{@{}L{4.2cm}L{4.0cm}L{2.2cm}L{3.4cm}L{2.0cm}L{7.6cm}@{}}
\caption{Selected non-parametric implied-volatility and option-price
interpolation methods.  Soft-penalty methods are approximate even when reported
violations are negligible.}
\label{tab:nonparametric-surface-interpolation-literature}\\
\toprule
Method & Representation & No static arb. & Native inputs & Off-grid & Scope and qualification \\
\midrule
\endfirsthead
\multicolumn{6}{c}{\tablename\ \thetable\ -- continued}\\
\toprule
Method & Representation & Static arb. & Native inputs & Off-grid & Scope and qualification \\
\midrule
\endhead
\midrule
\multicolumn{6}{r}{Continued on next page}\\
\endfoot
\bottomrule
\endlastfoot
\auditnew{\citet{Kahale2004}}
  & \auditnew{Piecewise analytic call/IV interpolation}
  & \auditnew{S$\Omega$ (conditional)} & \auditnew{Arbitrage-free discrete quotes}
  & \auditnew{$\infty$ for $T\leq T_{\max}$}
  & \auditnew{Convex-envelope construction on the bounded maturity domain; no extension beyond the last input expiry.} \\
\citet{fengler2009}
  & Constrained spline & S$\Omega$ per expiry; SG calendar & Arbitrary & $\infty$ on fitted strike domain
  & Strike constraints hold over each fitted spline; calendar constraints only at knots and input maturities. \\
\citet{AndreasenHuge}
  & One-step implicit FD; basis for DLV & SG numerically; S$\infty$ for exact local-vol model & Arbitrary quotes & $\infty$
  & Sequential local-vol calibration; the finite-difference scheme preserves monotonicity and convexity on its grid. \\
\citet{Orosi2015}
  & Thin-plate regression splines/QP & U & Arbitrary & $\infty$
  & Quadratic-program constraints; their continuum scope was not verified. \\
\citet{AckererTagasovskaVatter2020}
  & Neural network & A & Irregular daily panel & $\infty$
  & Static conditions enter through soft penalties; violation rates depend on the penalty weight. \\
\auditnew{\citet{BenderThiel2020}}
  & \auditnew{Reference-model distortion of call prices}
  & \auditnew{S$\infty$ (conditional)} & \auditnew{Arbitrage-free discrete quotes}
  & \auditnew{$\infty$}
  & \auditnew{Smooth strike interpolation and arbitrage-free maturity interpolation.} \\
\citet{ZhangLiZhang2021}
  & PCA/VAE/samples plus neural surface & A & $14\times11=154$ & $\infty$
  & Calendar, butterfly, and wing conditions enter as penalties; occasional spikes remain. \\
\citet{ChataignerEtAl2022}
  & Constrained GP or penalized NN & GP: S$\Omega$; NN: A & Arbitrary & $\infty$
  & Nodal inequalities imply monotonicity and convexity throughout the GP's bounded piecewise-multilinear domain.  The NN and fitted SSVI alternatives use penalties. \\
\auditnew{\citet{Deschatres2024}}
  & \auditnew{Variance-space cubic/B-spline QP} & \auditnew{A}
  & \auditnew{Arbitrary; 14,500 quotes in the SPX example} & \auditnew{$\infty$}
  & \auditnew{Calendar constraints are linear, but butterfly constraints are linearized and the QP is iterated.  Under the present convention this is approximate rather than a single exact admissible fit.} \\
\citet{YangChenShuHospedales2025}
  & Hypernetwork plus compact neural surface & A & 9 observations & $\infty$
  & Reconstructs index-option IV surfaces from sparse option quotes.  Calendar and density conditions are finite-grid auxiliary losses; the paper explicitly states that they provide no hard guarantee. \\
\auditnew{\citet{GononJacquierWiedemann2025}}
  & \auditnew{Graph neural operator} & \auditnew{A}
  & \auditnew{Irregular intraday panels} & \auditnew{$\infty$ on its output domain}
  & \auditnew{Calendar and butterfly conditions are finite-difference penalty losses evaluated at sampled coordinates; the paper describes the output as virtually arbitrage-free, not structurally invariant.} \\
\citet{Schlenkrich2025}
  & Piecewise shifted-Black/Bachelier model & S$\infty$ per expiry; calendar N & Arbitrary & $\infty$ per expiry
  & Whole-strike construction for one terminal distribution; calendar consistency is not addressed. \\
\citet{QinYangCheFeng2026}
  & Piecewise-quadratic calls/RN marginals & S$\infty$ in strike at quoted expiries & Arbitrage-free discrete quotes & $\infty$ in strike only
  & Explicit sequential completion of risk-neutral marginals; no interpolation between expiries. \\
\citet{SANOS2026}
  & Smooth non-parametric DLV & S$\infty$ & Arbitrary; about 1,000 in example & $\infty$
  & Bid--ask-aware smooth surface with positivity as the essential parameter constraint. \\
\end{longtable}
\end{landscape}

\begin{landscape}
\begin{table}[p]
\centering
\fontsize{7}{7.8}\selectfont
\setlength{\tabcolsep}{2.5pt}
\renewcommand{\arraystretch}{1.06}
\caption{PCA, autoencoder, and other latent representations of option surfaces.
``Off-grid'' counts options not represented by native model coordinates;
external interpolation is not credited to the representation.
  An off-grid count of
$\infty$ means that the representation can be evaluated on a continuum, subject
to its stated domain and tail assumptions.
}
\label{tab:surface-representation-literature}
\begin{tabular}{@{}L{4.3cm}L{4.2cm}L{3.6cm}L{2.7cm}L{3.0cm}L{6.0cm}@{}}
\toprule
Method & Representation & Native option coordinates & Off-grid & No static arb. & Qualification \\
\midrule
\citet{SkiadopoulosHodgesClewlow2000}
  & Direct-IV PCA & Not reported; three maturity groups & 0 & N & PCA factors do not preserve static inequalities. \\
\citet{ContVol}
  & Functional/direct-IV PCA & Not reported & Any number within the fitted domain & N & Functional interpolation does not make the PCA cone invariant. \\
\auditnew{\citet{FenglerHardleVilla2003}}
  & \auditnew{Common PCA of direct IV changes} & \auditnew{$6\times6=36$: moneyness 0.85--1.10 at 1, 2, 3, 6, 9, and 12 months (DAX)}
  & \auditnew{0} & \auditnew{N}
  & \auditnew{Common factors summarize smile dynamics but do not define an invariant no-arbitrage state space.} \\
\citet{Hafner2004}
  & Four parametric IV factors: level, slope, curvature, and term slope & Variable daily DAX transactions; 200-point diagnostic grid & $\infty$ on its bounded domain & N & The daily regression
   family is not constrained.  Dates failing finite-grid hedge, bull-spread, butterfly, or calendar tests are discarded ex post. \\
\citet{PanigirtzoglouSkiadopoulos2004}
  & Separate implied-distribution PCA by expiry & $3\times21=63$: 21 deltas at 1, 3, and 6 months & 0 & PDF: S$\infty$ per expiry; CDF: N; joint surface: N & The PDF construction is strict per expiry; neither construction enforces calendar consistency. \\
\citet{DellaportasMijatovic2014}
  & Parametric SABR encoding & SABR triplet for one expiry & Any strike in fitted interval & U & SABR may violate arbitrage conditions in the extreme wings. \\
\citet{wiese2019simulate}
  & DLV, PCA, and GAN & $8\times4=32$ & 0 & SG & Finite DLV lattice. \\
\auditnew{\citet{AvellanedaEtAl2020}}
  & \auditnew{PCA of normalized IV-surface deformations} & \auditnew{$8\times7=56$ per name: 8 maturities and 7 deltas}
  & \auditnew{0} & \auditnew{N}
  & \auditnew{A data-driven factor representation; no static-arbitrage invariant decoder is imposed.} \\
\auditnew{\citet{BergeronEtAl2021}}
  & \auditnew{VAE representation of IV surfaces} & \auditnew{$8\times5=40$: 8 maturities and 5 deltas (FX)}
  & \auditnew{Any (delta, maturity) in the domain via the pointwise decoder} & \auditnew{N}
  & \auditnew{A hands-off latent representation and generator, without a hard static-arbitrage architecture.} \\
\citet{wiese2021multiassetspotoptionmarket}
  & DLV, constrained autoencoder, and flow & $9\times4=36$ per asset & 0 & SG & Decoder is constrained on the finite lattice. \\
\citet{ZhangLiZhang2021}
  & PCA/VAE/samples plus neural reconstruction & 154 & $\infty$ & A & Reconstruction uses static-arbitrage penalties. \\
\citet{AndresBoumezouedJourdain2023}
  & Four-parameter SSVI state & 384 (Euro Stoxx 50); 1,944 (SPX, before filtering) & $\infty$ & S$\infty$ & Admissible SSVI parameter conditions hold for each state. \\
\citet{Cohen02012023}
  & Affine factors of normalized call prices & 46 options on an irregular rolling $(\tau,m)$ lattice & 0 & SG & Linear no-static-arbitrage inequalities define a factor polytope. \\
\citet{ContVuletic2023}
  & Direct-IV PCA/factors plus path reweighting & Fixed IV grid & 0 & Raw PCA: N; output: A & Arbitrage scenarios are downweighted, not excluded structurally. \\
\citet{ChoudharyJaimungalBergeron2024}
  & Functional PCA and neural SDE & $17\times11=187$ per asset & Any number within the fitted domain & N & Static-arbitrage violations are measured only ex post. \\
\citet{VuleticCont2024}
  & Direct-IV GAN plus penalties/reweighting & $10\times8=80$ & 0 & A & Approximate output correction. \\
\citet{JinAgarwal2025}
  & Direct-IV conditional diffusion & $9\times9=81$ & 0 & A & Penalty-based. \\
\citet{GolubitskyDavisonCaraman2026}
  & VAE latent representation & $2\times11\times17=374$ & 0 & A & Static constraints are approximate. \\
\citet{BrooksEtAl2026}
  & Arbitrage-penalized VAE and latent flow & $32\times16=512$ & 0 & A; 90.8\% pass all tests & Empirical pass rate, not a guarantee. \\
\citet{HaoJi2026}
  & Conditional diffusion plus SAAM refinement & $15\times17=255$ & Arbitrary target coordinates & A & SAAM uses residual regularization and reduces measured violations on a finite diagnostic grid; it is not an admissible parameterization. \\
\citet{HanEtAl2026}
  & Joint return/direct-IV adaptive sequential diffusion & $11\times9=99$ IVs & 0 & A & Fine-tuning penalizes calendar-, call-, and butterfly-spread violations; reported violations are nearly zero, not structurally excluded. \\
\citet{YangWangZhangLi2026}
  & Joint stock-return/option-IV diffusion across equities & $11\times9=99$ option IVs per stock & 0 & A & Transfers dynamics to unseen stocks on the same fixed grid. \\
\OwnRowsRepresentation
\bottomrule
\end{tabular}
\end{table}
\end{landscape}

\begin{landscape}
\begin{table}[p]
\centering
\fontsize{7}{7.8}\selectfont
\setlength{\tabcolsep}{2pt}
\renewcommand{\arraystretch}{1.08}
\caption{Machine-learning methods for underlying-model inverse problems.
These methods infer or optimize parameters of a pre-specified pricing model;
they do not make the observed option market itself the dynamic state.  Static
arbitrage properties are inherited from the pricing model, subject to decoded
parameters remaining admissible.}
\label{tab:inverse-calibration-literature}
\begin{tabular}{@{}L{4.5cm}L{5.0cm}L{5.2cm}L{4.2cm}L{5.2cm}@{}}
\toprule
Method & Underlying pricing family & Role of ML & Online inverse step & Why separate from a market-state model \\
\midrule
\citet{Hernandez2016}
  & Hull--White/G2++ examples & Direct inverse map trained on model-generated examples & One network evaluation & Maps quotes to parameters of a prescribed model. \\
\auditnew{\citet{DimitroffRoederFries2018}}
  & \auditnew{Heston}
  & \auditnew{Convolutional network maps volatility grids to model parameters}
  & \auditnew{One network evaluation}
  & \auditnew{Amortized inversion into a prescribed pricing family.} \\
\citet{BayerStemper2018}
  & Rough Bergomi and related rough-volatility models & Learns the forward parameter-to-IV map & Classical Levenberg--Marquardt optimization through the surrogate & The NN accelerates a conventional model calibration. \\
\citet{HorvathMuguruzaTomas2019,BayerEtAl2019DeepCalibration}
  & Broad class including rough-volatility models & Learns a fast forward pricing/IV map & Classical numerical or Bayesian inversion of the surrogate & Explicitly retains the two-step pricing-model-plus-calibration architecture. \\
\auditnew{\citet{LiuEtAl2019Calibration}}
  & \auditnew{Heston and Bates}
  & \auditnew{Learns a forward parameter-to-observable map}
  & \auditnew{Differential evolution through the surrogate}
  & \auditnew{The learned map accelerates conventional parameter calibration.} \\
\citet{Itkin2019}
  & Black--Scholes examples & Compares direct and surrogate calibration & Direct map or numerical solver & Studies smooth Greeks, soft arbitrage penalties, and inverse-map stability. \\
\citet{RoederDimitroff2020}
  & Volatility-model examples used in deep-calibration benchmarks & Compares a direct IV-to-parameter map with the two-step forward surrogate & None for the direct version; solver for the indirect version & Both outputs are parameters of an exogenously specified pricing model. \\
\auditnew{\citet{CuchieroKhosrawiTeichmann2020}}
  & \auditnew{Local stochastic-volatility models}
  & \auditnew{Neural leverage function}
  & \auditnew{Monte Carlo SGD; adversarial optimization only for robust bid--ask calibration}
  & \auditnew{The generator remains an underlying-first LSV model.} \\
\citet{RosenbaumZhang2021}
  & Quadratic rough Heston & Neural pricing/calibration engine based on a multifactor approximation & Optimization against the learned engine & The calibrated state is a rough-volatility parameter vector. \\
\citet{vae_ning2023arbitrage}
  & Time-inhomogeneous CTMC; additive Gaussian-mixture and CGMY/KoBoL L\'evy models & VAE/PCA learns the cross-date distribution of already calibrated parameters & Regularized inverse calibration is done first, date by date; CTMC maturities are fitted sequentially & The VAE generates model parameters, and prices are recovered through the model's characteristic function. \\
\citet{SridiBilokon2023}
  & Heston & Neural pricing surrogate & Differential evolution and Nelder--Mead & Heston supplies the risk-neutral structure. \\
\citet{ZhangAmiciMorandotti2024}
  & Heston & Differential network learns prices and sensitivities & Gradient calibration through the learned pricer & Heston supplies the risk-neutral structure. \\
\bottomrule
\end{tabular}
\end{table}
\end{landscape}

\begin{landscape}
\fontsize{7}{7.8}\selectfont
\setlength{\tabcolsep}{2.2pt}
\renewcommand{\arraystretch}{1.06}
\begin{longtable}{@{}L{3.7cm}L{3.4cm}L{2.8cm}L{2.2cm}L{4.1cm}L{1.5cm}L{5.3cm}@{}}
\caption{Option-surface dynamics under the historical measure $\P$ and direct
market-model dynamics under $\Q$.  The $\Q$ column excludes underlying-first
pricing families such as Heston, local volatility, and Bates.  ``Multi-step
$\P$ path'' records whether the cited work actually reports a simulated or
forecast path spanning more than one transition, rather than whether its model
could in principle be iterated.  ``Implemented'' refers to the stated dynamic
construction, not merely static preprocessing.}
\label{tab:surface-dynamics-literature}\\
\toprule
Method & Dynamic state/model & Static no-arb. under $\P$ & Multi-step $\P$ path & Direct market-model dynamics under $\Q$ & Implemented & Qualification \\
\midrule
\endfirsthead
\multicolumn{7}{c}{\tablename\ \thetable\ -- continued}\\
\toprule
Method & Dynamic state/model & Static arb. under $\P$ & Multi-step $\P$ path & Direct market-model dynamics under $\Q$ & Implemented & Qualification \\
\midrule
\endhead
\midrule
\multicolumn{7}{r}{Continued on next page}\\
\endfoot
\bottomrule
\endlastfoot
\auditnew{\citet{Schonbucher1999}}
  & \auditnew{One fixed-strike term structure of implied volatilities}
  & \auditnew{---} & \auditnew{---} & \auditnew{Exact drift restrictions across maturities for one strike}
  & \auditnew{MC implementation strategy only}
  & \auditnew{The setup initially permits a generic finite option set, but the developed discrete and continuous constructions vary maturity at one strike.  Independent dynamics across strikes at the same maturity are explicitly left for future research; this is not a full-surface model.} \\
\citet{SkiadopoulosHodgesClewlow2000}
  & Direct-IV PCA factors & N & No & --- & Yes & No static invariant state space. \\
\citet{ContVol}
  & Functional/direct-IV PCA factors & N & No & --- ($\P$ only) & Empirical PCA only & No constructed admissible $\Q$ dynamics. \\
\auditnew{\citet{ContDaFonseca2002Dynamics}}
  & \auditnew{Functional PCA factors of direct IV surfaces} & \auditnew{N}
  & \auditnew{No} & \auditnew{---} & \auditnew{Yes ($\P$)}
  & \auditnew{Empirical factor dynamics; distinct from the three-author stochastic-surface paper and not a constructed risk-neutral model.} \\
\auditnew{\citet{FenglerHardleVilla2003}}
  & \auditnew{Common PCA factors of direct IV changes} & \auditnew{N}
  & \auditnew{No} & \auditnew{---} & \auditnew{Yes ($\P$)}
  & \auditnew{Empirical CPCA of observed daily surfaces.  The factors are proposed as a starting point for VaR Monte Carlo, but the paper does not report synthetic multi-step IV paths.  There is no invariant static-arbitrage state space.} \\
\citet{Hafner2004}
  & Four IV factors with correlated OU dynamics & N; screened sample only & No (1-, 5-, and 10-day forecasts) & A/local & Yes & The fitted OU model violates the global drift condition; the $\Q$ construction is local and approximate. \\
\citet{PanigirtzoglouSkiadopoulos2004}
  & Separate distribution factors by constant maturity & PDF: S$\infty$ per expiry; CDF: N & Yes (14 days) & N for a joint surface & Yes & Neither construction enforces calendar consistency. \\
\auditnew{\citet{FenglerHardleMammen2007}}
  & \auditnew{Semiparametric IV-factor surface} & \auditnew{N}
  & \auditnew{No} & \auditnew{---} & \auditnew{Yes ($\P$)}
  & \auditnew{Estimates and forecasts historical surface factors without a hard static-arbitrage invariant.} \\
\citet{SchweizerWissel2008}
  & Price level and local implied-volatility field & S$\infty$, theoretical & --- & S$\infty$: solvability theorem and explicit admissible coefficient examples & No & One fixed expiry and all strikes; not a finite option lattice or full maturity surface. \\
\auditnew{\citet{SchweizerWisselTerm2008}}
  & \auditnew{Term structure for one fixed payoff} & \auditnew{S$\infty$, theoretical}
  & \auditnew{No} & \auditnew{Strict theoretical drift/existence results}
  & \auditnew{No}
  & \auditnew{Existence for fixed-strike calls and a power payoff; not a full-surface model.} \\
\auditnew{\citet{CarmonaNadtochiy2009}}
  & \auditnew{Local-volatility surface as option-market codebook} & \auditnew{---}
  & \auditnew{---} & \auditnew{Drift and spot-consistency characterization}
  & \auditnew{Static empirical fit only}
  & \auditnew{Direct HJM-style codebook dynamics; general model existence is not proved.} \\
\auditnew{\citet{TangentLevy}}
  & \auditnew{Tangent L\'evy option-market codebook} & \auditnew{---}
  & \auditnew{---} & \auditnew{S$\infty$ under the consistency conditions}
  & \auditnew{Theory}
  & \auditnew{Constructs direct risk-neutral dynamics for all strikes and maturities from admissible tangent L\'evy characteristics.} \\
\citet{DellaportasMijatovic2014}
  & Single-expiry SABR parameter process & U & No (1--3-day point forecasts) & --- & Yes ($\P$) & No joint surface; extreme-wing arbitrage is not excluded. \\
\auditnew{\citet{KallsenKruhner2015}}
  & \auditnew{Time-inhomogeneous L\'evy option codebook} & \auditnew{---}
  & \auditnew{---} & \auditnew{S$\infty$: necessary/sufficient conditions and existence}
  & \auditnew{No numerical implementation}
  & \auditnew{Constructive existence and uniqueness from stated building blocks, but a theoretical HJM framework rather than a calibrated simulator.} \\
\auditnew{\citet{CarmonaMaNadtochiy2017}}
  & \auditnew{Finite-dimensional tangent L\'evy specifications} & \auditnew{---}
  & \auditnew{No} & \auditnew{S$\infty$ under model admissibility}
  & \auditnew{Yes}
  & \auditnew{Risk-neutral HJM drift with PCA-estimated volatility; eight-day $\Q$ paths.} \\
\pagebreak[4]
\citet{wiese2019simulate}
  & DLV PCA/GAN state & SG & Yes & --- & Yes & Finite-grid historical dynamics. \\
\auditnew{\citet{CaoChenHull2020}}
  & \auditnew{Neural regression of one-day IV changes} & \auditnew{N}
  & \auditnew{No} & \auditnew{---} & \auditnew{Yes ($\P$)}
  & \auditnew{Expected changes depend on return, delta/moneyness, maturity, and VIX; no arbitrage constraint.} \\
\citet{wiese2021multiassetspotoptionmarket}
  & Constrained DLV latent flow & SG & Yes & --- & Yes & Finite-grid historical dynamics. \\
\citet{ZhangLiZhang2021}
  & LSTM latent evolution and neural reconstruction & A & No (next-day only) & --- & Yes & Static-arbitrage penalties only. \\
\auditnew{\citet{removedrift}}
  & \auditnew{General $\P$-simulator plus an OCE/entropy path-density model}
  & \auditnew{Inherited from the input simulator}
  & \auditnew{Input-dependent}
  & \auditnew{Numerical near-martingale measure for a finite trading universe}
  & \auditnew{Yes}
  & \auditnew{Reweights an input simulator; with frictions, drift is targeted within transaction costs for selected instruments.} \\
\citet{MedvedevWang2022}
  & Direct-IV LSTM/ConvLSTM & N & No (multihorizon forecasts) & --- & Yes ($\P$) & No static-arbitrage constraint. \\
\citet{AndresBoumezouedJourdain2023}
  & Path-dependent SSVI factors & S$\infty$ for every simulated snapshot & Yes & No market-measure dynamics & Yes ($\P$) & Static arbitrage freedom does not imply dynamic no-arbitrage. \\
\citet{Cohen02012023}
  & Rolling normalized-call lattice, affine factors, and neural SDE & SG under $\P$: exact for model SDE & Yes & --- (not implemented) & Yes ($\P$ only) & The drift points inward and the normal diffusion vanishes at every face, so the continuous-time factor process remains in the finite-lattice no-static-arbitrage polytope $\P$-a.s.  The coordinates float with fixed $(\tau,m)$ rather than tracking fixed contracts.  The HJM expression is only a factor-selection criterion and residual diagnostic, not an implemented $\Q$ model. \\
\citet{ContVuletic2023}
  & Direct-IV factors plus path reweighting & A & Yes & --- & Yes & Reweighting reduces violations. \\
\citet{ChoudharyJaimungalBergeron2024}
  & Functional PCA and neural SDE & N & Yes & --- & Yes & Violations are measured only ex post. \\
\citet{VuleticCont2024}
  & Direct-IV GAN dynamics & A & No (next-day only) & --- & Yes & Penalties and reweighting. \\
\citet{JinAgarwal2025}
  & Conditional diffusion on IV grids & A & No (next-day only) & --- & Yes & Penalty-based. \\
\citet{BrooksEtAl2026}
  & Arbitrage-penalized VAE and latent flow matching & A & No & --- & Yes & Only 90.8\% of generated surfaces pass all reported tests. \\
\citet{HaoJi2026}
  & Conditional diffusion plus SAAM refinement & A & Yes (five one-minute steps) & --- & Yes & The daily experiment is next-day only; finite-grid violations remain. \\
\citet{HanEtAl2026}
  & Joint daily return and direct-IV grid via adaptive sequential diffusion & A & Yes (reported three-day path) & --- & Yes ($\P$) & Each step conditions on the updated generated history; hedging uses one-day conditional scenarios. \\
\citet{YangWangZhangLi2026}
  & Joint stock-return/option-IV diffusion across equities & A & No (next-day only) & --- & Yes & Jointly generates the next-day stock return and $11\times9$ option-IV increment; static option constraints are approximate. \\
\citet{FrancoisEtAl2026}
  & Five-factor IV surface jointly with SPX returns & A & No (five-day terminal scenarios) & --- & Yes ($\P$) & Rare empirical violations, not an invariant state space. \\
\citet{BucheggerGonon2026}
  & Latent diffusion for returns and IV surfaces & A & Yes (30-day trajectories) & --- & Yes ($\P$) & Penalty-based arbitrage awareness; violations remain. \\
\OwnRowsDynamics
\end{longtable}
\end{landscape}

\bibliographystyle{plainnat}
\bibliography{refs}

\begin{thebibliography}{88}
\providecommand{\natexlab}[1]{#1}
\providecommand{\url}[1]{\texttt{#1}}
\expandafter\ifx\csname urlstyle\endcsname\relax
  \providecommand{\doi}[1]{doi: #1}\else
  \providecommand{\doi}{doi: \begingroup \urlstyle{rm}\Url}\fi

\bibitem[Ackerer et~al.(2020)Ackerer, Tagasovska, and
  Vatter]{AckererTagasovskaVatter2020}
Damien Ackerer, Natasa Tagasovska, and Thibault Vatter.
\newblock Deep smoothing of the implied volatility surface.
\newblock In \emph{Advances in Neural Information Processing Systems},
  volume~33, pages 11552--11563, 2020.

\bibitem[Andreasen and Huge(2011)]{AndreasenHuge}
Jesper Andreasen and Brian Huge.
\newblock Volatility interpolation.
\newblock \emph{Risk}, 24\penalty0 (3):\penalty0 76--79, 2011.

\bibitem[Andr{\`e}s et~al.(2026)Andr{\`e}s, Boumezoued, and
  Jourdain]{AndresBoumezouedJourdain2023}
Herv{\'e} Andr{\`e}s, Alexandre Boumezoued, and Benjamin Jourdain.
\newblock The implied volatility surface (also) is path-dependent.
\newblock \emph{Quant. Finance}, 26\penalty0 (7):\penalty0 1017--1047, 2026.
\newblock \doi{10.1080/14697688.2026.2637739}.

\bibitem[Avellaneda et~al.(2020)Avellaneda, Healy, Papanicolaou, and
  Papanicolaou]{AvellanedaEtAl2020}
Marco Avellaneda, Brian~F. Healy, Andrew Papanicolaou, and George Papanicolaou.
\newblock {PCA} for implied volatility surfaces.
\newblock \emph{J. Financial Data Sci.}, 2\penalty0 (2):\penalty0 85--109,
  2020.
\newblock \doi{10.3905/jfds.2020.1.032}.

\bibitem[Bayer and Stemper(2018)]{BayerStemper2018}
Christian Bayer and Benjamin Stemper.
\newblock Deep calibration of rough stochastic volatility models, 2018.
\newblock URL \url{https://arxiv.org/abs/1810.03399}.

\bibitem[Bayer et~al.(2019)Bayer, Horvath, Muguruza, Stemper, and
  Tomas]{BayerEtAl2019DeepCalibration}
Christian Bayer, Blanka Horvath, Aitor Muguruza, Benjamin Stemper, and Mehdi
  Tomas.
\newblock On deep calibration of (rough) stochastic volatility models, 2019.
\newblock URL \url{https://arxiv.org/abs/1908.08806}.

\bibitem[Bender and Thiel(2020)]{BenderThiel2020}
Christian Bender and Matthias Thiel.
\newblock Arbitrage-free interpolation of call option prices.
\newblock \emph{Stat. Risk Model.}, 37\penalty0 (1--2):\penalty0 55--78, 2020.
\newblock \doi{10.1515/strm-2018-0026}.

\bibitem[Bergeron et~al.(2021)Bergeron, Fung, Hull, and
  Poulos]{BergeronEtAl2021}
Maxime Bergeron, Nicholas Fung, John Hull, and Zissis Poulos.
\newblock Variational autoencoders: A hands-off approach to volatility, 2021.

\bibitem[Bergomi(2005)]{DysLitBergomi2005}
Lorenzo Bergomi.
\newblock Smile dynamics {II}.
\newblock \emph{Risk}, 18\penalty0 (10):\penalty0 67--73, 2005.
\newblock URL
  \url{https://www.risk.net/derivatives/equity-derivatives/1500225/smile-dynamics-ii}.

\bibitem[Bergomi(2008)]{bergomi2008smile}
Lorenzo Bergomi.
\newblock Smile dynamics {III}.
\newblock \emph{Risk}, 21\penalty0 (10):\penalty0 90--96, 2008.

\bibitem[Bergomi and Guyon(2012)]{bergomi2012stochastic}
Lorenzo Bergomi and Julien Guyon.
\newblock Stochastic volatility's orderly smiles.
\newblock \emph{Risk}, 25\penalty0 (5):\penalty0 60--66, 2012.

\bibitem[Bloch(2013)]{Bloch2013MixVar}
Daniel~Alexandre Bloch.
\newblock Arbitrage-free mix var volatility surfaces and applications.
\newblock SSRN working paper 2197185, 2013.
\newblock URL \url{https://papers.ssrn.com/abstract=2197185}.

\bibitem[Brace et~al.(1997)Brace, Gatarek, and Musiela]{BGM}
Alan Brace, Dariusz Gatarek, and Marek Musiela.
\newblock The market model of interest rate dynamics.
\newblock \emph{Mathematical Finance}, 7\penalty0 (2):\penalty0 127--155, 1997.
\newblock \doi{10.1111/1467-9965.00028}.

\bibitem[Brooks et~al.(2026)Brooks, Bajalica, Liu, and
  Ben~Tahar]{BrooksEtAl2026}
Oscar Brooks, Dusica Bajalica, Yating Liu, and Imen Ben~Tahar.
\newblock Latent flow matching for arbitrage-aware implied volatility surface
  generation, 2026.

\bibitem[Buchegger and Gonon(2026)]{BucheggerGonon2026}
Dominik~Manuel Buchegger and Lukas Gonon.
\newblock Arbitrage-aware multi-step forecasting of implied volatility
  surfaces: Modelling surface trajectories using latent diffusion, 2026.
\newblock URL \url{https://arxiv.org/abs/2608.22478}.

\bibitem[Buehler(2006)]{buehler2006consistent}
Hans Buehler.
\newblock Consistent variance curve models.
\newblock \emph{Finance and Stochastics}, 10\penalty0 (2):\penalty0 178--203,
  2006.
\newblock \doi{10.1007/s00780-006-0008-2}.

\bibitem[Buehler and Ryskin(2015)]{DLV}
Hans Buehler and Evgeny Ryskin.
\newblock Discrete local volatility for large time steps (extended version).
\newblock SSRN working paper, 2015.
\newblock URL \url{https://papers.ssrn.com/abstract=2642630}.

\bibitem[Buehler et~al.(2022)Buehler, Murray, Pakkanen, and Wood]{removedrift}
Hans Buehler, Phillip Murray, Mikko~S. Pakkanen, and Ben Wood.
\newblock Deep hedging: Learning to remove the drift.
\newblock \emph{Risk}, February 2022.
\newblock URL
  \url{https://www.risk.net/cutting-edge/banking/7932226/deep-hedging-learning-to-remove-the-drift}.

\bibitem[Buehler et~al.(2026{\natexlab{a}})Buehler, Horvath, and
  Kratsios]{DYSANOS2026}
Hans Buehler, Blanka Horvath, and Anastasis Kratsios.
\newblock {DYSANOS}: Generative dynamic smooth arbitrage-free non-parametric
  option surfaces, 2026{\natexlab{a}}.
\newblock URL \url{https://arxiv.org/abs/2608.12587}.

\bibitem[Buehler et~al.(2026{\natexlab{b}})Buehler, Horvath, Kratsios, Limmer,
  and Saqur]{SANOS2026}
Hans Buehler, Blanka Horvath, Anastasis Kratsios, Yannick Limmer, and Raeid
  Saqur.
\newblock {SANOS}: Smooth strictly arbitrage-free non-parametric option
  surfaces, January 2026{\natexlab{b}}.
\newblock URL \url{https://arxiv.org/abs/2601.11209}.
\newblock submitted.

\bibitem[Cao et~al.(2020)Cao, Chen, and Hull]{CaoChenHull2020}
Jay Cao, Jacky Chen, and John Hull.
\newblock A neural network approach to understanding implied volatility
  movements.
\newblock \emph{Quant. Finance}, 20\penalty0 (9):\penalty0 1405--1413, 2020.
\newblock \doi{10.1080/14697688.2020.1750679}.

\bibitem[Carmona and Nadtochiy(2009)]{CarmonaNadtochiy2009}
Ren{\'e} Carmona and Sergey Nadtochiy.
\newblock Local volatility dynamic models.
\newblock \emph{Finance Stoch.}, 13\penalty0 (1):\penalty0 1--48, 2009.
\newblock \doi{10.1007/s00780-008-0078-4}.

\bibitem[Carmona and Nadtochiy(2012)]{TangentLevy}
Ren{\'e} Carmona and Sergey Nadtochiy.
\newblock Tangent l{\'e}vy market models.
\newblock \emph{Finance Stoch.}, 16:\penalty0 63--104, 2012.
\newblock \doi{10.1007/s00780-011-0158-8}.

\bibitem[Carmona et~al.(2017)Carmona, Ma, and
  Nadtochiy]{CarmonaMaNadtochiy2017}
Ren{\'e} Carmona, Yi~Ma, and Sergey Nadtochiy.
\newblock Simulation of implied volatility surfaces via tangent {L\'evy}
  models.
\newblock \emph{SIAM J. Financial Math.}, 8\penalty0 (1):\penalty0 171--213,
  2017.
\newblock \doi{10.1137/15M1015510}.

\bibitem[Chataigner et~al.(2021)Chataigner, Cousin, Cr{\'e}pey, Dixon, and
  Gueye]{ChataignerEtAl2022}
Marc Chataigner, Areski Cousin, St{\'e}phane Cr{\'e}pey, Matthew Dixon, and
  Djibril Gueye.
\newblock Beyond surrogate modeling: Learning the local volatility via shape
  constraints.
\newblock \emph{SIAM J. Financial Math.}, 12\penalty0 (3):\penalty0 SC58--SC69,
  2021.
\newblock \doi{10.1137/20M1381538}.

\bibitem[Choudhary et~al.(2024)Choudhary, Jaimungal, and
  Bergeron]{ChoudharyJaimungalBergeron2024}
Vedant Choudhary, Sebastian Jaimungal, and Maxime Bergeron.
\newblock {FuNVol}: A multi-asset implied volatility market simulator using
  functional principal components and neural {SDEs}.
\newblock \emph{Quant. Finance}, 24\penalty0 (8):\penalty0 1077--1103, 2024.
\newblock \doi{10.1080/14697688.2024.2396977}.

\bibitem[Ciliberti et~al.(2008)Ciliberti, Bouchaud, and
  Potters]{DysLitCilibertiBouchaudPotters2008}
Stefano Ciliberti, Jean-Philippe Bouchaud, and Marc Potters.
\newblock Smile dynamics---a theory of the implied leverage effect.
\newblock \url{https://arxiv.org/abs/0809.3375}, 2008.
\newblock arXiv:0809.3375.

\bibitem[Cohen et~al.(2023)Cohen, Reisinger, and Wang]{Cohen02012023}
Samuel~N. Cohen, Christoph Reisinger, and Sheng Wang.
\newblock Arbitrage-free neural-{SDE} market models.
\newblock \emph{Appl. Math. Finance}, 30\penalty0 (1):\penalty0 1--46, 2023.
\newblock \doi{10.1080/1350486X.2023.2257217}.

\bibitem[Cont and da~Fonseca(2002)]{ContDaFonseca2002Dynamics}
Rama Cont and Jos{\'e} da~Fonseca.
\newblock Dynamics of implied volatility surfaces.
\newblock \emph{Quant. Finance}, 2\penalty0 (1):\penalty0 45--60, 2002.
\newblock \doi{10.1088/1469-7688/2/1/304}.

\bibitem[Cont and Vuleti{\'c}(2023)]{ContVuletic2023}
Rama Cont and Milena Vuleti{\'c}.
\newblock Simulation of arbitrage-free implied volatility surfaces.
\newblock \emph{Appl. Math. Finance}, 30\penalty0 (2):\penalty0 94--121, 2023.
\newblock \doi{10.1080/1350486X.2023.2277960}.

\bibitem[Cont et~al.(2002)Cont, da~Fonseca, and Durrleman]{ContVol}
Rama Cont, Jos{\'e} da~Fonseca, and Valdo Durrleman.
\newblock Stochastic models of implied volatility surfaces.
\newblock \emph{Econ. Notes}, 31\penalty0 (2):\penalty0 361--377, 2002.
\newblock \doi{10.1111/1468-0300.00090}.

\bibitem[Corbetta et~al.(2019)Corbetta, Cohort, Laachir, and
  Martini]{CorbettaEtAl2019}
Jacopo Corbetta, Pierre Cohort, Ismail Laachir, and Claude Martini.
\newblock Robust calibration and arbitrage-free interpolation of {SSVI} slices.
\newblock \emph{Decis. Econ. Finance}, 42\penalty0 (2):\penalty0 665--677,
  2019.

\bibitem[Cuchiero et~al.(2020)Cuchiero, Khosrawi, and
  Teichmann]{CuchieroKhosrawiTeichmann2020}
Christa Cuchiero, Wahid Khosrawi, and Josef Teichmann.
\newblock A generative adversarial network approach to calibration of local
  stochastic volatility models.
\newblock \emph{Risks}, 8\penalty0 (4):\penalty0 101, 2020.
\newblock \doi{10.3390/risks8040101}.

\bibitem[Dellaportas and Mijatovi{\'c}(2014)]{DellaportasMijatovic2014}
Petros Dellaportas and Aleksandar Mijatovi{\'c}.
\newblock Arbitrage-free prediction of the implied volatility smile, 2014.

\bibitem[Desch{\^a}tres(2024)]{Deschatres2024}
Fabrice Desch{\^a}tres.
\newblock Convex volatility interpolation.
\newblock SSRN working paper 4831218, 2024.
\newblock URL \url{https://papers.ssrn.com/abstract=4831218}.
\newblock Revised 2026.

\bibitem[Dimitroff et~al.(2018)Dimitroff, R{\"o}der, and
  Fries]{DimitroffRoederFries2018}
Georgi Dimitroff, Dirk R{\"o}der, and Christian~P. Fries.
\newblock Volatility model calibration with convolutional neural networks.
\newblock SSRN working paper 3252432, 2018.
\newblock URL \url{https://papers.ssrn.com/abstract=3252432}.

\bibitem[El~Amrani et~al.(2021)El~Amrani, Jacquier, and
  Martini]{DysLitElAmrani2021}
Mehdi El~Amrani, Antoine Jacquier, and Claude Martini.
\newblock Short communication: Dynamics of symmetric {SSVI} smiles and implied
  volatility bubbles.
\newblock \emph{SIAM Journal on Financial Mathematics}, 12\penalty0
  (2):\penalty0 SC1--SC15, 2021.
\newblock \doi{10.1137/20M136089X}.

\bibitem[Feller(1968)]{feller-vol-1}
William Feller.
\newblock \emph{An introduction to probability theory and its applications.
  {V}ol. {I}}.
\newblock Third edition. John Wiley \& Sons Inc., New York, 1968.

\bibitem[Fengler(2009)]{fengler2009}
Matthias~R. Fengler.
\newblock Arbitrage-free smoothing of the implied volatility surface.
\newblock \emph{Quant. Finance}, 9\penalty0 (4):\penalty0 417--428, 2009.
\newblock \doi{10.1080/14697680802595585}.

\bibitem[Fengler et~al.(2003)Fengler, H{\"a}rdle, and
  Villa]{FenglerHardleVilla2003}
Matthias~R. Fengler, Wolfgang~K. H{\"a}rdle, and Christophe Villa.
\newblock The dynamics of implied volatilities: A common principal components
  approach.
\newblock \emph{Rev. Deriv. Res.}, 6\penalty0 (3):\penalty0 179--202, 2003.
\newblock \doi{10.1023/B:REDR.0000004823.77464.2D}.

\bibitem[Fengler et~al.(2007)Fengler, H{\"a}rdle, and
  Mammen]{FenglerHardleMammen2007}
Matthias~R. Fengler, Wolfgang~K. H{\"a}rdle, and Enno Mammen.
\newblock A semiparametric factor model for implied volatility surface
  dynamics.
\newblock \emph{J. Financial Econometrics}, 5\penalty0 (2):\penalty0 189--218,
  2007.
\newblock \doi{10.1093/jjfinec/nbm005}.

\bibitem[Fran{\c{c}}ois et~al.(2026)Fran{\c{c}}ois, Galarneau-Vincent,
  Gauthier, and Godin]{FrancoisEtAl2026}
Pascal Fran{\c{c}}ois, R{\'e}mi Galarneau-Vincent, Genevi{\`e}ve Gauthier, and
  Fr{\'e}d{\'e}ric Godin.
\newblock Joint dynamics for the underlying asset and its implied volatility
  surface: A new methodology for option risk management.
\newblock \emph{J. Futures Mark.}, 46\penalty0 (3):\penalty0 545--561, 2026.
\newblock \doi{10.1002/fut.70068}.

\bibitem[Fukasawa(2022)]{DysLitFukasawa2022}
Masaaki Fukasawa.
\newblock On asymptotically arbitrage-free approximations of the implied
  volatility.
\newblock \emph{Frontiers of Mathematical Finance}, 1\penalty0 (4):\penalty0
  525--537, 2022.
\newblock \doi{10.3934/fmf.2022006}.

\bibitem[Fukasawa and Gatheral(2022)]{DysLitFukasawaGatheral2022}
Masaaki Fukasawa and Jim Gatheral.
\newblock A rough {SABR} formula.
\newblock \emph{Frontiers of Mathematical Finance}, 1\penalty0 (1):\penalty0
  81--97, 2022.
\newblock \doi{10.3934/fmf.2021003}.

\bibitem[Gambara and Teichmann(2021)]{DysLitGambaraTeichmann2021}
Matteo Gambara and Josef Teichmann.
\newblock Consistent recalibration models and deep calibration.
\newblock \url{https://arxiv.org/abs/2006.09455}, 2021.
\newblock arXiv:2006.09455, version 3; first posted 2020.

\bibitem[Gatheral(2004)]{Gatheral2004SVI}
Jim Gatheral.
\newblock A parsimonious arbitrage-free implied volatility parameterization
  with application to the valuation of volatility derivatives.
\newblock Presentation at Global Derivatives \& Risk Management, Madrid, 2004.

\bibitem[Gatheral and Jacquier(2014)]{GatheralJacquier2014}
Jim Gatheral and Antoine Jacquier.
\newblock Arbitrage-free {SVI} volatility surfaces.
\newblock \emph{Quant. Finance}, 14\penalty0 (1):\penalty0 59--71, 2014.
\newblock \doi{10.1080/14697688.2013.819986}.

\bibitem[Golubitsky et~al.(2026)Golubitsky, Davison, and
  Caraman]{GolubitskyDavisonCaraman2026}
Letitia Golubitsky, Matt Davison, and Alex Caraman.
\newblock Unsupervised learning of equity market dynamics: A variational
  autoencoder approach.
\newblock SSRN working paper 6726284, 2026.
\newblock URL \url{https://papers.ssrn.com/abstract=6726284}.

\bibitem[Guo et~al.(2016)Guo, Jacquier, Martini, and
  Neufcourt]{GuoJacquierMartiniNeufcourt2016}
Gaoyue Guo, Antoine Jacquier, Claude Martini, and Leo Neufcourt.
\newblock Generalized arbitrage-free {SVI} volatility surfaces.
\newblock \emph{SIAM J. Financial Math.}, 7\penalty0 (1):\penalty0 619--641,
  2016.
\newblock \doi{10.1137/120900320}.

\bibitem[Hafner(2004)]{Hafner2004}
Reinhold Hafner.
\newblock \emph{Stochastic Implied Volatility: A Factor-Based Model}, volume
  545 of \emph{Lecture Notes in Economics and Mathematical Systems}.
\newblock Springer Berlin Heidelberg, 2004.
\newblock ISBN 978-3-540-22183-8.
\newblock \doi{10.1007/978-3-642-17117-8}.
\newblock URL \url{https://books.google.com/books?id=67OE8cKXWBQC}.

\bibitem[Han et~al.(2026)Han, Zhang, Torres, Acero, and Xu]{HanEtAl2026}
Yinbin Han, Jack~Yuxiang Zhang, Manuel Torres, Fernando Acero, and Renyuan Xu.
\newblock Diffusion models for dynamic volatility surface generation and
  data-driven hedging, 2026.
\newblock URL \url{https://arxiv.org/abs/2609.13402}.

\bibitem[Hao and Ji(2026)]{HaoJi2026}
Lifeng Hao and Shaolin Ji.
\newblock Decoupled probabilistic forecasting and arbitrage-aware refinement of
  implied volatility surfaces, 2026.
\newblock URL \url{https://arxiv.org/abs/2607.29220}.

\bibitem[Hendriks and Martini(2019)]{HendriksMartini2019}
Sebas Hendriks and Claude Martini.
\newblock The extended {SSVI} volatility surface.
\newblock \emph{J. Comput. Finance}, 22\penalty0 (5):\penalty0 25--39, 2019.
\newblock \doi{10.21314/JCF.2019.365}.

\bibitem[Hernandez(2016)]{Hernandez2016}
Andres Hernandez.
\newblock Model calibration with neural networks.
\newblock SSRN working paper 2812140, 2016.
\newblock URL \url{https://papers.ssrn.com/abstract=2812140}.

\bibitem[Horvath et~al.(2019)Horvath, Muguruza, and
  Tomas]{HorvathMuguruzaTomas2019}
Blanka Horvath, Aitor Muguruza, and Mehdi Tomas.
\newblock Deep learning volatility, 2019.
\newblock URL \url{https://arxiv.org/abs/1901.09647}.

\bibitem[Itkin(2019)]{Itkin2019}
Andrey Itkin.
\newblock Deep learning calibration of option pricing models: Some pitfalls and
  solutions, 2019.
\newblock URL \url{https://arxiv.org/abs/1906.03507}.

\bibitem[Jin and Agarwal(2025)]{JinAgarwal2025}
Chen Jin and Ankush Agarwal.
\newblock Forecasting implied volatility surface with generative diffusion
  models, 2025.

\bibitem[Kahal{\'e}(2004)]{Kahale2004}
Nabil Kahal{\'e}.
\newblock An arbitrage-free interpolation of volatilities.
\newblock \emph{Risk}, 17\penalty0 (5):\penalty0 102--106, 2004.

\bibitem[Kallsen and Kr{\"u}hner(2015)]{KallsenKruhner2015}
Jan Kallsen and Paul Kr{\"u}hner.
\newblock On a {Heath--Jarrow--Morton} approach for stock options.
\newblock \emph{Finance Stoch.}, 19\penalty0 (3):\penalty0 583--615, 2015.
\newblock \doi{10.1007/s00780-015-0263-1}.

\bibitem[Ledoit et~al.(2002)Ledoit, Santa-Clara, and Yan]{ledoit2002relative}
Olivier Ledoit, Pedro Santa-Clara, and Shu Yan.
\newblock Relative pricing of options with stochastic volatility.
\newblock Working paper, Anderson Graduate School of Management, UCLA, 2002.
\newblock URL \url{https://escholarship.org/uc/item/7jp8f42t}.

\bibitem[Liu et~al.(2019)Liu, Borovykh, Grzelak, and
  Oosterlee]{LiuEtAl2019Calibration}
Shuaiqiang Liu, Anastasia Borovykh, Lech~A. Grzelak, and Cornelis~W. Oosterlee.
\newblock A neural network-based framework for financial model calibration.
\newblock \emph{J. Math. Ind.}, 9:\penalty0 9, 2019.
\newblock \doi{10.1186/s13362-019-0066-7}.

\bibitem[Lucic(2021)]{lucic2021normalizing}
Vladimir Lucic.
\newblock Normalizing volatility transforms and parameterization of volatility
  smile.
\newblock \emph{Available at SSRN 3835233}, 2021.

\bibitem[Martini and Mingone(2022)]{MartiniMingone2022}
Claude Martini and Arianna Mingone.
\newblock No arbitrage {SVI}.
\newblock \emph{SIAM J. Financial Math.}, 13\penalty0 (1):\penalty0 227--261,
  2022.

\bibitem[Medvedev and Wang(2022)]{MedvedevWang2022}
Nikita Medvedev and Zhiguang Wang.
\newblock Multistep forecast of the implied volatility surface using deep
  learning.
\newblock \emph{J. Futures Mark.}, 42\penalty0 (4):\penalty0 645--667, 2022.
\newblock \doi{10.1002/fut.22302}.

\bibitem[Mingone(2022)]{Mingone2022}
Arianna Mingone.
\newblock No arbitrage global parametrization for the {eSSVI} volatility
  surface.
\newblock \emph{Quant. Finance}, 22\penalty0 (12):\penalty0 2205--2217, 2022.

\bibitem[Ning et~al.(2023)Ning, Jaimungal, Zhang, and
  Bergeron]{vae_ning2023arbitrage}
Brian Ning, Sebastian Jaimungal, Xiaorong Zhang, and Maxime Bergeron.
\newblock Arbitrage-free implied volatility surface generation with variational
  autoencoders.
\newblock \emph{SIAM J. Financial Math.}, 14\penalty0 (4):\penalty0 1004--1027,
  2023.
\newblock \doi{10.1137/21M1443546}.

\bibitem[Orosi(2015)]{Orosi2015}
Gergely Orosi.
\newblock Arbitrage-free call option surface construction using regression
  splines.
\newblock \emph{Appl. Stoch. Models Bus. Ind.}, 31\penalty0 (4):\penalty0
  515--527, 2015.
\newblock \doi{10.1002/asmb.2045}.

\bibitem[Panigirtzoglou and Skiadopoulos(2004)]{PanigirtzoglouSkiadopoulos2004}
Nikolaos Panigirtzoglou and George Skiadopoulos.
\newblock A new approach to modeling the dynamics of implied distributions:
  Theory and evidence from the {S\&P 500} options.
\newblock \emph{J. Bank. Finance}, 28\penalty0 (7):\penalty0 1499--1520, 2004.
\newblock \doi{10.1016/S0378-4266(03)00127-4}.

\bibitem[Protter(2005)]{Protter2005}
Philip~E. Protter.
\newblock \emph{Stochastic Integration and Differential Equations}.
\newblock Springer, second edition, 2005.

\bibitem[Qin et~al.(2026)Qin, Yang, Che, and Feng]{QinYangCheFeng2026}
Hao Qin, Ruozhong Yang, Charlie~Y. Che, and Liming Feng.
\newblock Arbitrage-free multi-maturity risk-neutral marginals.
\newblock arXiv preprint arXiv:2607.06204, 2026.
\newblock URL \url{https://arxiv.org/abs/2607.06204}.

\bibitem[Roeder and Dimitroff(2020)]{RoederDimitroff2020}
Dirk Roeder and Georgi Dimitroff.
\newblock Volatility model calibration with neural networks: A comparison
  between direct and indirect methods, 2020.
\newblock URL \url{https://arxiv.org/abs/2007.03494}.

\bibitem[Roos(2025)]{Roos2025}
Thomas Roos.
\newblock Simple, flexible, analytic, arbitrage free option price
  interpolation.
\newblock SSRN working paper 5215592, 2025.
\newblock URL \url{https://papers.ssrn.com/abstract=5215592}.
\newblock Revised 2026.

\bibitem[Rosenbaum and Zhang(2021)]{RosenbaumZhang2021}
Mathieu Rosenbaum and Jianfei Zhang.
\newblock Deep calibration of the quadratic rough {Heston} model, 2021.
\newblock URL \url{https://arxiv.org/abs/2107.01611}.

\bibitem[Schlenkrich(2025)]{Schlenkrich2025}
Sebastian Schlenkrich.
\newblock A piece-wise model for vanilla option pricing.
\newblock SSRN working paper 5214566, 2025.
\newblock URL \url{https://papers.ssrn.com/abstract=5214566}.

\bibitem[Sch{\"o}nbucher(1999)]{Schonbucher1999}
Philipp~J. Sch{\"o}nbucher.
\newblock A market model for stochastic implied volatility.
\newblock \emph{Philos. Trans. R. Soc. Lond. A}, 357\penalty0 (1758):\penalty0
  2071--2092, 1999.
\newblock \doi{10.1098/rsta.1999.0418}.

\bibitem[Schweizer and Wissel(2008{\natexlab{a}})]{SchweizerWissel2008}
Martin Schweizer and Johannes Wissel.
\newblock Arbitrage-free market models for option prices: The multi-strike
  case.
\newblock \emph{Finance Stoch.}, 12\penalty0 (4):\penalty0 469--505,
  2008{\natexlab{a}}.
\newblock \doi{10.1007/s00780-008-0068-6}.

\bibitem[Schweizer and Wissel(2008{\natexlab{b}})]{SchweizerWisselTerm2008}
Martin Schweizer and Johannes Wissel.
\newblock Term structures of implied volatilities: Absence of arbitrage and
  existence results.
\newblock \emph{Math. Finance}, 18\penalty0 (1):\penalty0 77--114,
  2008{\natexlab{b}}.
\newblock \doi{10.1111/j.1467-9965.2007.00323.x}.

\bibitem[Skiadopoulos et~al.(2000)Skiadopoulos, Hodges, and
  Clewlow]{SkiadopoulosHodgesClewlow2000}
George~S. Skiadopoulos, Stewart~D. Hodges, and Les Clewlow.
\newblock The dynamics of the {S\&P 500} implied volatility surface.
\newblock \emph{Rev. Deriv. Res.}, 3\penalty0 (3):\penalty0 263--282, 2000.
\newblock \doi{10.1023/A:1009642705121}.

\bibitem[Sridi and Bilokon(2023)]{SridiBilokon2023}
Abir Sridi and Paul Bilokon.
\newblock Applying deep learning to calibrate stochastic volatility models,
  2023.
\newblock URL \url{https://arxiv.org/abs/2309.07843}.

\bibitem[Vuleti{\'c} and Cont(2024)]{VuleticCont2024}
Milena Vuleti{\'c} and Rama Cont.
\newblock {VolGAN}: A generative model for arbitrage-free implied volatility
  surfaces.
\newblock \emph{Appl. Math. Finance}, 31\penalty0 (4):\penalty0 203--238, 2024.
\newblock \doi{10.1080/1350486X.2025.2471317}.

\bibitem[Wiedemann et~al.(2025)Wiedemann, Jacquier, and
  Gonon]{GononJacquierWiedemann2025}
Ruben Wiedemann, Antoine Jacquier, and Lukas Gonon.
\newblock Operator deep smoothing for implied volatility.
\newblock In \emph{International Conference on Learning Representations}, 2025.
\newblock URL \url{https://arxiv.org/abs/2406.11520}.

\bibitem[Wiese et~al.(2019)Wiese, Bai, Wood, and Buehler]{wiese2019simulate}
Magnus Wiese, Lianjun Bai, Ben Wood, and Hans Buehler.
\newblock Deep hedging: Learning to simulate equity option markets, 2019.
\newblock URL \url{https://arxiv.org/abs/1911.01700}.

\bibitem[Wiese et~al.(2021)Wiese, Wood, Pachoud, Korn, Buehler, Murray, and
  Bai]{wiese2021multiassetspotoptionmarket}
Magnus Wiese, Ben Wood, Alexandre Pachoud, Ralf Korn, Hans Buehler, Phillip
  Murray, and Lianjun Bai.
\newblock Multi-asset spot and option market simulation, 2021.
\newblock URL \url{https://arxiv.org/abs/2112.06823}.

\bibitem[Wissel(2007)]{Wissel2007}
Johannes Wissel.
\newblock Arbitrage-free market models for option prices.
\newblock NCCR FINRISK Working Paper 428, NCCR FINRISK, ETH Z{\"u}rich, 2007.

\bibitem[Yang et~al.(2026)Yang, Wang, Zhang, and Li]{YangWangZhangLi2026}
Mingzhi Yang, Sheng Wang, Chao Zhang, and Ruikun Li.
\newblock Universal diffusion models for implied volatility surfaces: Learning
  shared dynamics across stocks, 2026.

\bibitem[Yang et~al.(2025)Yang, Chen, Shu, and
  Hospedales]{YangChenShuHospedales2025}
Yongxin Yang, Wenqi Chen, Chao Shu, and Timothy~M. Hospedales.
\newblock {HyperIV}: Real-time implied volatility smoothing.
\newblock In \emph{Proceedings of the 42nd International Conference on Machine
  Learning}, volume 267 of \emph{Proceedings of Machine Learning Research},
  pages 70550--70564. PMLR, 2025.
\newblock URL \url{https://proceedings.mlr.press/v267/yang25d.html}.

\bibitem[Zhang et~al.(2024)Zhang, Amici, and
  Morandotti]{ZhangAmiciMorandotti2024}
Chen Zhang, Giovanni Amici, and Marco Morandotti.
\newblock Calibrating the {Heston} model with deep differential networks, 2024.
\newblock URL \url{https://arxiv.org/abs/2407.15536}.

\bibitem[Zhang et~al.(2023)Zhang, Li, and Zhang]{ZhangLiZhang2021}
Wenyong Zhang, Lingfei Li, and Gongqiu Zhang.
\newblock A two-step framework for arbitrage-free prediction of the implied
  volatility surface.
\newblock \emph{Quant. Finance}, 23\penalty0 (1):\penalty0 21--34, 2023.
\newblock \doi{10.1080/14697688.2022.2135454}.

\end{thebibliography}

\end{document}